\documentclass{article}
\usepackage[utf8]{inputenc}
\usepackage{fullpage}

\usepackage{amsfonts}
\usepackage{amsmath}
\usepackage{amsthm}
\usepackage{mathtools}

\usepackage{amssymb}
\usepackage{xcolor}
\usepackage{enumerate}
\usepackage{thm-restate}

\usepackage{upgreek}

\usepackage{graphicx}
\graphicspath{ {./images/} }

\usepackage{microtype}

\usepackage[colorlinks]{hyperref}
\usepackage[capitalize, noabbrev]{cleveref}

\newtheorem{theorem}{Theorem}[section]

\newtheorem{lemma}{Lemma}[section]
\newtheorem{definition}{Definition}[section]

\newtheorem{claim}{Claim}[section]
\newtheorem{observation}{Observation}[section]

\newcommand{\R}{\mathbb{R}}
\newcommand{\Rp}{\mathbb{R}_+}
\newcommand{\Rnn}{\mathbb{R}_{\ge 0}}
\newcommand{\fO}{\mathcal{O}}
\newcommand{\fW}{\mathcal{W}}
\newcommand{\rr}{\mathrm{r}}

\newcommand{\cupdot}{\mathbin{\mathaccent\cdot\cup}}

\newcommand{\TT}{\mathcal{T}}
\newcommand{\CENT}{\mathtt{cent}}
\newcommand{\HH}{\mathcal{H}}
\newcommand{\A}{\mathcal{A}}
\newcommand{\B}{\mathcal{B}}

\newcommand{\OPT}{\mathtt{OPT}}
\newcommand{\cost}{\mathtt{cost}}
\newcommand{\depth}{\mathtt{depth}}

\renewcommand{\root}{\mathtt{root}}
\newcommand{\dist}{d}

\newcommand{\Path}{\mathtt{Path}}
\newcommand{\LCA}{\mathtt{LCA}}
\newcommand{\CC}{\mathbb{C}}

\newcommand{\sforest}{\mathfrak{F}}
\newcommand{\stree}{\mathfrak{T}}

\renewcommand{\setminus}{-}

\newcommand{\al}{\upalpha}

\date{}

\title{Fast approximation of search trees on trees with centroid trees}
\author{
	Benjamin Aram Berendsohn\thanks{Institut f\"ur Informatik, Freie Universit\"at Berlin, Germany. Email: \texttt{beab@zedat.fu-berlin.de}.  {Supported by DFG grant KO 6140/1-1.}}  \and 
	Ishay Golinsky\thanks{Blavatnik School of Computer Science, Tel Aviv University, Israel. Email: \texttt{ishayg@mail.tau.ac.il}.  {Supported by ISF
grant no.~1595-19 and the Blavatnik Family Foundation.}}
	\and
	Haim Kaplan\thanks{Blavatnik School of Computer Science, Tel Aviv University, Israel. Email: \texttt{haimk@tau.ac.il}.  {Supported by ISF
grant no.~1595-19 and the Blavatnik Family Foundation.}}
	\and
	L\'aszl\'o Kozma\thanks{Institut f\"ur Informatik, Freie Universit\"at Berlin, Germany. Email: \texttt{laszlo.kozma@fu-berlin.de}.  {Supported by DFG grant KO 6140/1-1.}}}

\begin{document}
\maketitle

\begin{abstract}
Search trees on trees (STTs) generalize the fundamental binary search tree (BST) data structure: in STTs the underlying search space is an arbitrary tree, whereas in BSTs it is \emph{a path}. An optimal BST of size $n$ can be computed for a given distribution of queries in $\fO(n^2)$ time [Knuth, Acta Inf.\ 1971] and \emph{centroid} BSTs provide a nearly-optimal alternative, computable in $\fO(n)$ time [Mehlhorn, SICOMP 1977]. 

By contrast, optimal STTs are not known to be computable in polynomial time, and the fastest constant-approximation algorithm runs in $\fO(n^3)$ time [Berendsohn, Kozma, SODA 2022]. Centroid trees can be defined for STTs analogously to BSTs, and they have been used in a wide range of algorithmic applications. In the unweighted case (i.e., for a uniform distribution of queries), the centroid tree can be computed in $\fO(n)$ time [Brodal, Fagerberg, Pedersen, \"{O}stlin, ICALP 2001; Della Giustina, Prezza, Venturini, SPIRE 2019]. These algorithms, however, do not readily extend to the weighted case. Moreover, no approximation guarantees were previously known for centroid trees in either the unweighted or weighted cases. 

In this paper we revisit centroid trees in a general, weighted setting, and we settle both the algorithmic complexity of constructing them, and the quality of their approximation. For constructing a weighted centroid tree, we give an \emph{output-sensitive} $\fO(n\log{h}) \subseteq \fO(n \log{n})$ time algorithm, where $h$ is the height of the resulting centroid tree. If the weights are of polynomial complexity, the running time is  $\fO(n\log\log{n})$. We show these bounds to be optimal, in a general decision tree model of computation. For approximation, we prove that the cost of a centroid tree is at most \emph{twice} the optimum, and this guarantee is best possible, both in the weighted and unweighted cases. We also give tight, fine-grained bounds on the approximation-ratio for bounded-degree trees and on the approximation-ratio of more general $\al$-centroid trees. 
\end{abstract}

\section{Introduction}\label{sec1}
Search trees on trees (STTs) are a far-reaching generalization of binary search trees (BSTs), modeling the exploration of tree-shaped search spaces. Given an undirected tree $\TT$, an STT on $\TT$ is a tree rooted at an arbitrary vertex $r$ of $\TT$, with subtrees built recursively on the components resulting after removing $r$ from $\TT$, see Figure~\ref{fig1} for an example. BSTs correspond to the special case where the underlying tree $\TT$ is a \emph{path}. 

STTs and, more generally, search trees on graphs arise in several different contexts and have been studied under different names: \emph{tubings}~\cite{carr2006coxeter}, \emph{vertex rankings}~\cite{deogun, bodlaender98, EvenSmorodinsky}, \emph{ordered colorings}~\cite{katchalski1995ordered}, \emph{elimination trees}~\cite{liu1990role, pothen1990, AspvallHeggernes, bodlaender95}. STTs have been crucial in many algorithmic applications, e.g., in pattern matching and counting~\cite{Ferragina13, kociumaka2014, gagie2015}, cache-oblivious data structures~\cite{bender2006, FerraginaV16}, tree clustering~\cite{frederickson1983}, geometric visibility~\cite{GuibasHLST86}, planar point location~\cite{GT98}, distance oracles~\cite{distoracle}. They arise in matrix factorization (e.g., see~\cite[\S\,12]{duff2017}), and have also been related to the competitive ratio in certain online hitting set problems~\cite{EvenSmorodinsky}. 

Similarly to the setting of BSTs, a natural goal is to find an STT in which the expected depth of a vertex is as small as possible; we refer to such a tree as an \emph{optimal tree}, noting that it is not necessarily unique. This optimization task can be studied both for the uniform probability distribution over the vertices, and for the more general case of an arbitrary distribution given as input. We refer to the first as the \emph{unweighted} and the second as the \emph{weighted} problem.

For BSTs, both the unweighted and the weighted problems are well-understood. In the unweighted case, a simple balanced binary tree achieves the optimum. In the weighted case, an optimal tree on $n$ vertices can be found in time $\fO(n^2)$ by Knuth's algorithm~\cite{knuth_optimum}, a textbook example of dynamic programming. No faster algorithm is known in general, although Larmore's algorithm~\cite{Larmore87} achieves better bounds under certain regularity assumptions on the weights; for example, if the probability assigned to each vertex is $\Omega(1/n)$, then the optimum can be found in time $\fO(n^{1.591})$.

By contrast, the complexity of computing an optimal STT is far less understood. Even in the unweighted case, no polynomial-time algorithm is known, and the problem is not known to be NP-hard even with arbitrary weights. Recently, a PTAS was given for the weighted problem~\cite{BK22}, but its running time for obtaining a $(1+\varepsilon)$-approximation of the optimal STT is $\fO(n^{1+2/\varepsilon})$, which is prohibitive for reasonably small values of $\varepsilon$. Note that the apparently easier problem of minimizing the \emph{maximum depth} of a vertex, i.e., computing the \emph{treedepth} of a tree, can be solved in linear time by Sch\"{a}ffer's algorithm~\cite{schaeffer1989}, and treedepth itself has many algorithmic applications, e.g., see~\cite[\S\,6,7]{sparsity}.

\paragraph{Centroid trees.}

Given the relatively high cost of computing optimal binary search trees, research has turned already half a century ago to efficient approximations. Mehlhorn has shown~\cite{mehlhorn1975nearly, mehlhorn77abest} that a simple BST that can be computed in $\fO(n)$ time closely approximates the optimum. More precisely, both the optimum cost and the cost of the obtained tree are in $[{H}/{\log(3)}, H+1]$, where $H$ is the binary entropy of the input distribution.\footnote{All logarithms in this paper are base $2$.} Alternatively, the cost can be upper bounded by $\OPT + \log{(\OPT)}+ \log{e}$, where $\OPT$ is the cost of the optimal tree. Observe that this means that the approximation ratio gets arbitrarily close to $1$ as $\OPT$ goes to infinity. \footnote{Results for BSTs are sometimes presented in a more general form, where the input distribution also accounts for \emph{unsuccessful searches}, i.e., it may assign non-zero probabilities to the \emph{gaps} between neighboring vertices and outside the two extremes. Extending such a model to STTs is straightforward, but perhaps less natural in the case of general trees, we therefore omit it for the sake of simplicity, and consider only successful searches.}

\begin{figure}
\centering
\includegraphics[width=5in]{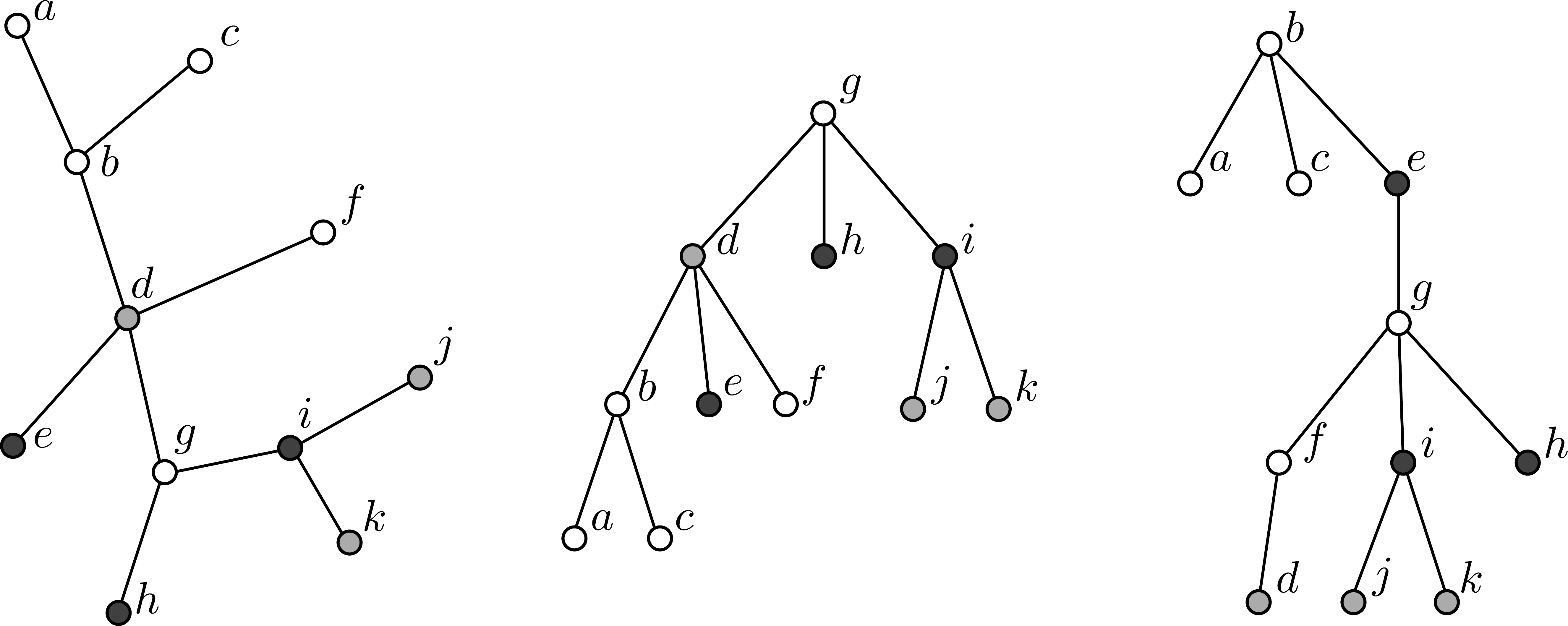}
\caption{(\emph{Left.}) Tree $\TT$. (\emph{Middle.}) Centroid tree of $\TT$. (\emph{Right.}) A different STT on $\TT$. Colors indicate weights (probabilities), $w(e)=w(h)=w(i)=0.15$, $w(d)=w(j)=w(k)=0.10$, and all other vertices have weight $0.05$. Observe that the centroid tree is (in this example) unique.}
\label{fig1}
\end{figure}

The BST that achieves the above guarantees is built by recursively picking roots such as to make the weights of the left and right subtrees ``as equal as possible''. This is a special case of a \emph{centroid tree}, defined as follows. Given a tree $\TT$, a \emph{centroid} of $\TT$ is a vertex whose removal from $\TT$ results in components with weight \emph{at most half} of the total weight of $\TT$. A centroid tree is built by iteratively finding a centroid and recursing on the components resulting after its removal. See Figure~\ref{fig1} for an example.

The fact that an (unweighted) centroid always exists was already shown in the 19-th century by C.\ Jordan~\cite{Jordan1869}. We sketch the easy, constructive argument that also shows the existence of a weighted centroid: start at an arbitrary vertex of $\TT$ and, as long as the current vertex is not a centroid, move one edge in the direction of the component with largest weight. It is not hard to see that the procedure succeeds, visiting each vertex at most once. 

A straightforward implementation of the above procedure finds an unweighted centroid tree in $\fO(n \log{n})$ time. This running time has been improved to $\fO(n)$ by carefully using data structures~\cite{BrodalFPO01, Giustina}.  The run-time guarantees however, do not readily generalize from the unweighted to the weighted setting. Intuitively, the difficulty lies in the fact that in the weighted case, the removal of a centroid vertex may split the tree in a very unbalanced way, leaving up to $n-1$ vertices in one component. Thus, a naive recursive approach will take $\Theta(n^2)$ time in the worst case.

Most algorithmic applications of STTs, including those mentioned before, rely on centroid trees. It is therefore surprising that nothing appears to be known about how well the centroid tree approximates the optimal STT in either the unweighted or weighted cases. In this paper we prove that the centroid tree is a $2$-approximation of the optimal STT, and that the factor $2$ is, in general, best possible, both in the unweighted and weighted settings. As our main result, we also show a more precise bound on the approximation ratio of centroid trees, in terms of the maximum degree of the underlying tree $\TT$.\footnote{In their recent paper on dynamic STTs, Bose, Cardinal, Iacono, Koumoutsos, and Langerman~\cite{Bose20} remark that the ratio between the costs of the centroid- and optimal trees may be unbounded. In light of our results, this observation is erroneous. It is true, however, that a centroid tree built using the uniform distribution may be far from the optimum w.r.t.\ a different distribution.}

\medskip

Before stating our results, we need a few definitions. Consider an undirected, unrooted tree $\TT$ given as input, together with a \emph{weight function} $w:V(\TT) \rightarrow \mathbb{R}_{\geq 0}$. For convenience, for any subgraph $\HH$ of $\TT$, we denote $w(\HH) = \sum_{x \in V(\HH)}{w(x)}$. (To interpret the weights as probabilities, we need the condition $w(\TT)=1$. It is, however, often convenient to relax this requirement and allow arbitrary non-negative weights, which is the approach we take in this paper.)

A \emph{search tree} on $\TT$ is a rooted tree $T$ with vertex set $V(\TT)$ whose root is an arbitrary vertex $r \in V(\TT)$. The children of $r$ in $T$ are the roots of search trees built on the connected components of the forest $\TT \setminus {r}$. A tree consisting of a single vertex admits only itself as a search tree. (See Figure~\ref{fig1}.) It follows from the definition that for all $x$, the subtree $T_x$ of $T$ rooted at $x$ induces a connected subgraph $\TT[V(T_x)]$ of $\TT$, and moreover, $T_x$ is a search tree on $\TT[V(T_x)]$. 

The \emph{cost} of a search tree $T$ on $\TT$ is $\cost_w(T) = \sum_{x \in V(T)} {w(x) \cdot \depth_T{(x)}}$, where the depth of the root is taken to be $1$. The \emph{optimum cost} $\OPT(\TT,w)$ is the minimum of $\cost_w(T)$ over all search trees $T$ of $\TT$. 

A vertex $v \in V(\TT)$ is a \emph{centroid} if for all components ${\HH}$ of $\TT \setminus v$, we have $w(\HH) \leq w(\TT) / 2$. A search tree $T$ of $\TT$ is a \emph{centroid tree} if vertex $x$ is a centroid of $\TT[V(T_x)]$ for all $x \in V(\TT)$. In general, the centroid tree is not unique, and centroid trees of the same tree can have different costs.\footnote{Consider, for instance the two different centroid trees of a path on four vertices, with weights $(0.2, 0.3, 0.2, 0.3)$.} We denote by $\CENT(\TT,w)$ the \emph{maximum} cost of a centroid tree of $(\TT,w)$, with weight function $w$.

We can now state our approximation guarantee for centroid trees.

\begin{restatable}{theorem}{restatethma}\label{thm1}
Let $\TT$ be a tree, $w:V(\TT) \rightarrow \mathbb{R}_{\geq 0}$, and $m = w(\TT)$.
Then
$$\CENT(\TT,w) \leq 2 \cdot \OPT(\TT,w)-m.$$
\end{restatable}

We show that this result is optimal, including in the additive term. Moreover, the constant factor $2$ cannot be improved even for unweighted instances.

\begin{restatable}{theorem}{restatethmb} \label{thm2}   

\begin{enumerate}[(i)]
    \item For every $\varepsilon>0$ there is a sequence of instances $(\TT_n,w_n)$ with $w_n(\TT_n) = 1$, and for every centroid tree $C_n$ of $(\TT_n,w_n)$
    $$\cost_{w_n}(C_n)\geq 2\cdot\OPT(\TT_n,w_n) - 1 - \varepsilon.$$
    \item There is a sequence of instances $(\TT_n,w_n)$, where $w_n$ is the uniform distribution on $V(\TT_n)$, 
    and for every centroid tree $C_n$ of $(\TT_n, w_n)$
    $$\lim_{n\rightarrow\infty}\frac{{\cost}_{w_n}(C_n)}{\OPT(\TT_n,w_n)} = 2.$$
\end{enumerate}

In both cases $\displaystyle\lim_{n\rightarrow\infty}\OPT(\TT_n,w_n)=\infty$.

\end{restatable}

Note that the fact that $\lim_{n\rightarrow\infty}\OPT(\TT_n,w_n)=\infty$ in \Cref{thm2} establishes that the \emph{asymptotic} approximation ratio is $2$. By this we mean that every bound of the form $\CENT\leq c\cdot\OPT + o(\OPT)$ must have $c\geq 2$.

We next show a stronger guarantee when the underlying tree has bounded degree. 

\begin{restatable}{theorem}{restatethmd} \label{thm4}
Let $\TT$ be a tree, $w:V(\TT) \rightarrow \mathbb{R}_{\geq 0}$, and let $\Delta$ be the \emph{maximum degree} of $\TT$. Then
$$\CENT(\TT,w) \leq \left(2-\frac{1}{2^\Delta}\right) \cdot \OPT(\TT,w).$$
\end{restatable}

We complement this result by two lower bounds. The first establishes the tightness of the approximation ratio. The second shows a (slightly smaller) lower bound on the approximation ratio for instances where $\OPT$ is unbounded.

\begin{restatable}{theorem}{restatethme} \label{thm5}
Let $\Delta\geq3$ be integer.
\begin{enumerate}[(i)]
    \item
There is a sequence of instances $(\TT_n,w_n)$ such that $\TT_n$ has maximum degree at most $\Delta$, and for every centroid tree $C_n$ of $(\TT_n,w_n)$
    $$\lim_{n\rightarrow\infty}\frac{{\cost}_{w_n}(C_n)}{\OPT(\TT_n,w_n)} = 2 - \frac{1}{2^\Delta}.$$
    \item
    There is a sequence of instances $(\TT_n,w_n)$ such that $\TT_n$ has maximum degree at most $\Delta$,  $\displaystyle\lim_{n\rightarrow\infty}\OPT(\TT_n,w_n)=\infty$, $w_n(\TT_n) = 1$, and for every centroid tree $C_n$ of $(\TT_n,w_n)$
\begin{equation*}
\label{eq7}
    \cost_{w_n}(C_n) \geq \left(2-\frac{4}{2^\Delta}\right)\cdot\OPT(\TT_n,w_n) - 1.
\end{equation*}
    
    \end{enumerate}
\end{restatable}

We remark that Theorem~\ref{thm5}(i) does not exclude the possiblity of a bound of the form $\CENT \leq c \cdot \OPT + o(\OPT)$, where $c < 2-\frac{1}{2^\Delta}$, as here $\OPT(\TT_n,w_n)$ is bounded. Part (ii), however, establishes that a bound of the form $\CENT \leq c \cdot \OPT + o(\OPT)$ must have $c \geq 2-\frac{4}{2^\Delta}$. We leave open the problem of closing the gap in the asymptotic approximation ratio in terms of $\Delta$. 

\paragraph{Computing centroid trees.}
On the algorithmic side, we show that the weighted centroid tree can be computed in $\fO(n \log{n})$ time. Previously, the fastest known constant-approximation algorithm~\cite{BK22} took $\fO(n^3)$ time (similarly achieving an approximation ratio of $2$). The main step of our algorithm, finding the weighted centroid of a tree, is  achievable  in $\fO(\log{n})$ time, assuming that the underlying tree is stored in a \emph{top tree} data structure~\cite{AlstrupEtAl2005}. Iterating this procedure in combination with known algorithms for \emph{constructing} and \emph{splitting} top trees yields the algorithm that runs in $\fO(n \log{n})$ time. We also develop an improved, \emph{output-sensitive} algorithm, with running time $\fO(n \log{h})$, where $h$ is the height of the resulting centroid tree, yielding a running time $\fO(n \log\log{n})$ in the typical case when the height is $\fO(\log{n})$.

\begin{restatable}{theorem}{restateConstructHeight}\label{p:construct-height} 
	Let $\TT$ be a tree on $n$ vertices and $w$ be a weight function. We can compute a centroid tree of $(\TT,w)$ in time $\fO( n \log h )$, where $h$ is the height of the computed centroid tree.
\end{restatable}

One may ask whether the weighted centroid tree can be computed  in linear time, similarly to the unweighted centroid tree, or to the weighted centroid BST. We show that, assuming a general decision tree model of computation, this is not possible, and the algorithm of Theorem~\ref{p:construct-height} is optimal for all $n$ and $h$ (up to a constant factor). Our lower bound on the running time applies, informally, to any deterministic algorithm in which the input weights affect program flow only in the form of binary decisions, involving arbitrary computable functions. The model thus excludes using the weights for addressing memory, e.g., via hashing. 

More precisely, consider a tree $\TT$ on $n$ vertices. We say that a \emph{binary decision tree} $D_{\TT}$ \emph{solves $\TT$} for a class of weight functions $\mathcal{W}$ mapping $V(\TT)$ to $\mathbb{R}_{\geq 0}$, if the leaves of $D_{\TT}$ are search trees on $\TT$, every branching of $D_{\TT}$ is of the form ``$f(w) {\geq} 0?$'' for some computable function $f: \mathcal{W} \rightarrow \{-1,+1\}$, and for every weight function $w\in \mathcal{W}$, starting from the root of $D_{\TT}$ and following branchings down the tree, we reach a leaf $T$ of $D_{\TT}$ that is a valid centroid tree for $(\TT, w)$. The height of $D_{\TT}$ is then a lower bound on the worst-case running time. 

\begin{restatable}{theorem}{restatethmq}\label{thm:lb}
	Let $h \ge 3$ and $n \ge h+1$ be integers. Then there is a tree $\TT$ on at most $n$ vertices and a class $\fW$ of weight functions on $V(\TT)$ such that for every $w \in \fW$, every centroid tree of $(\TT,w)$ has height $h$, and every binary decision tree that solves $\TT$ for $\fW$ has height $\Omega(n \log h)$.
\end{restatable}

We can nonetheless improve the running time, when the weights are restricted in certain (natural) ways. We define the \emph{spread} $\sigma$ of a weight function $w$ as the ratio between the total weight $w(\TT)$, and the smallest \emph{non-zero} weight of a vertex. As we show, $\fO( n \log h ) \subseteq \fO( n \log \log {(\sigma+n)} )$ and therefore, when $\sigma \in n^{\fO(1)}$ (for instance, if the weights are integers stored in RAM words), we obtain a running time of $\fO(n\log\log{n})$.

When many vertices have zero weight, we obtain further improvements, e.g., if only $\fO(n/\log{n})$ of the weights are non-zero, we can compute a centroid tree in $\fO(n)$ time, even if the height $h$ is large. We defer the precise statement of these refined bounds and the discussion of their optimality to Section~\ref{sec5}.

\paragraph{Approximate centroid trees.} Finally, we consider the approximation guarantees of a generalized form of centroid trees. Let us call a vertex $v$ of a tree $\TT$ an $\al$-centroid, for $0 \leq \al \leq 1$, if $w(\HH) \leq \al \cdot w(\TT)$, for all components $\HH$ of $\TT \setminus v$.  
An $\al$-centroid tree is an STT in which every vertex $x$ is an $\al$-centroid of its subtree $\TT[V(T_x)]$. 

Observe that the standard centroid tree is a $\frac{1}{2}$-centroid tree, and all STTs are  $1$-centroid trees. Also note that an $\al$-centroid is a $\upbeta$-centroid for all $\upbeta \geq \al$ and that the existence of an $\al$-centroid is not guaranteed for $\al < \frac{1}{2}$ (consider a single edge with the two endpoints having the same weight). On the other hand, an $\al$-centroid for $\al < \frac{1}{2}$, if it exists, is unique, and therefore the $\al$-centroid tree is also unique. To see this, consider an $\al$-centroid $c$ that splits $\TT$ into components $\TT_1, \dots, \TT_k$. If an alternative $\al$-centroid $c'$ were in component $\TT_i$, then its removal would yield a component containing all vertices in $\TT \setminus V(\TT_i)$, of weight at least $(1-\al) \cdot w(\TT) > \al \cdot w(\TT)$.

Denote by $\CENT^\al(\TT,w)$ the maximum cost of an $\al$-centroid tree of $(\TT,w)$, or $0$ if no $\al$-centroid tree exists. We refine our guarantee from \Cref{thm1} to approximate centroid trees:

\begin{restatable}{theorem}{restatethmg}\label{thm7}
Let $\TT$ be a tree, $w:V(\TT) \rightarrow \mathbb{R}_{\geq 0}$, $m=w(\TT)$. We have 
\begin{align*}
(i)~~ \CENT^\al(\TT,w) &~\leq~ \frac{1}{1-\al} \cdot \OPT(\TT,w)-\frac{\al}{1-\al}m, & \mbox{for~} \al\in(0,1),\\
(ii)~~ \CENT^\al(\TT,w) &~\leq~ \frac{1}{2-3\al} \cdot \OPT(\TT,w)-\frac{3\al-1}{2-3\al}m, & \mbox{for~} \al\in\left[\frac{1}{3},\frac{1}{2}\right].
\end{align*}

\end{restatable}

Note that the second bound is a strengthening of the first when $\al<\frac{1}{2}$. In particular, for $\al\leq \frac{1}{3}$, it implies that an $\al$-centroid tree is \emph{optimal}, if it exists.

We show that the result is tight when $\al\geq \frac{1}{2}$ 
by proving a matching lower bound.

\begin{restatable}{theorem}{restatethmh} \label{thm8}   For every $\al\in[\frac{1}{2},1)$ there is a sequence of instances $(\TT_n,w_n)$ with $\displaystyle\lim_{n\rightarrow\infty}\OPT(\TT_n,w_n)=\infty$, $w_n(\TT_n) = 1$ and
    $$\CENT^\al(\TT_n,w_n)\geq \frac{1}{1-\al}\cdot \OPT(\TT_n,w_n) - \frac{\al}{1-\al}.$$
\end{restatable}

Note that if $\al>\frac{1}{2}$, we cannot prove such a lower bound for \emph{all} $\al$-centroid trees of $(\TT_n,w_n)$ (as in \Cref{thm2}), since a $\frac{1}{2}$-centroid tree exists and has stronger approximation guarantees according to \Cref{thm1}.

Finally, we argue that every optimal STT is a $\frac{2}{3}$-centroid tree. A special case of this result (for BSTs) was shown by Hirschberg, Larmore, and Molodowitch~\cite{HLM}, who also showed that the ratio $\frac{2}{3}$ is tight (in the special case of BSTs, and thus, also for STTs). 

\begin{restatable}{theorem}{restatethmoptc}\label{thm10}
Let $T$ be an optimal STT of $(\TT,w)$. Then, $T$ is a $\frac{2}{3}$-centroid tree of $(\TT,w)$.    
\end{restatable}

\paragraph{Structure of the paper.}
In Section~\ref{sec2} we state a number of technical results needed in the proofs. In Section~\ref{sec3} we prove the general upper and lower bounds on the approximation ratio of centroid trees (Theorems~\ref{thm1} and \ref{thm2}). In Section~\ref{sec4} we prove the fine-grained bounds on the approximation ratio of centroid trees (Theorems~\ref{thm4} and \ref{thm5}). Section~\ref{sec5} contains the algorithmic results (Theorem~\ref{p:construct-height} and extensions) and the lower bounds (Theorems~\ref{thm:lb} and extensions). Results on $\al$-centroids (Theorems~\ref{thm7}, \ref{thm8}, and \ref{thm10}) are proved in Section~\ref{sec6}. In Section~\ref{sec7} we conclude with open questions. 

\paragraph{Related work.}

Different models of searching in trees have also been considered, e.g., the one where we query edges instead of vertices~\cite{BenAsher, LaberNogueira, MozesOnak, OnakParys}, with connections to searching in posets~\cite{LinialSaks2, Heeringa}. In the edge-query setting, Cicalese, Jacobs, Laber and Molinaro~\cite{Cicalese3, Cicalese1} study the problem of minimizing the average search time of a vertex, and show this to be an NP-hard problem~\cite{Cicalese3}. They also show that an ``edge-centroid'' tree (in their terminology, a greedy algorithm) gives a $1.62$-approximation of the optimum~\cite{Cicalese1}. 

STTs generalize BSTs, therefore it is natural to ask to what extent the theory developed for BSTs can be extended to STTs. Defining a natural \emph{rotation} operation on STTs, Bose, Cardinal, Iacono, Koumoutsos, and Langerman~\cite{Bose20} develop an $O(\log\log{n})$ competitive dynamic STT, analogously to Tango BSTs~\cite{tango}. In a similar spirit, Berendsohn and Kozma~\cite{BK22} generalize Splay trees~\cite{SleatorTarjan1985} to STTs. The rotation operation on STTs naturally leads to the definition of \emph{tree associahedra}, a combinatorial structure that extends the classical associahedron defined over BSTs or other Catalan-structures. Properties of tree- and more general graph associahedra have been studied in~\cite{carr2006coxeter, devadoss2009realization, ceballos2015, Cardinal18, CPV21, Berendsohn2022}.

Searching in trees and graphs has also been motivated with applications,  including file system synchronisation~\cite{BenAsher, MozesOnak}, software testing~\cite{BenAsher, MozesOnak}, asymmetric communication protocols~\cite{LaberMolinaro}, VLSI layout~\cite{Leiserson80}, and assembly planning~\cite{Iyer1}. 
 
\section{Preliminaries}\label{sec2}

Given a graph $G$, we denote by $V(G)$ its set of vertices, by $E(G)$ its set of edges, and by $\mathbb{C}({G})$ its set of connected components. If $v \in V(G)$, denote by $N_G(v)$ the set of neighbors of $v$ in $G$, and $\deg_G(v) = |N_G(v)|$. For $S \subseteq V(G)$, denote by $G[S]$ the subgraph of $G$ induced by $S$, and for brevity, $G - v = G[V(G) \setminus \{v\}]$, and $G - S = G[V(G) \setminus S]$.

The following observation is straightforward.

\begin{observation}
\label{observation99}
    Let $T$ be a search tree on $\TT$, $w:V(\TT)\rightarrow\Rnn$, $m=w(\TT)$ and $r=\root(T)$. For each component $\HH\in\CC(\TT-r)$,  denote by $T_\HH$ the subtree of $T$ rooted at the unique child of $r$ in $\HH$. Then
    \begin{equation*}
        \cost_w(T) = m + \sum_{\mathclap{\HH\in\CC(\TT-r)}}~\cost_w(T_\HH) \geq m + \sum_{\mathclap{\HH\in\CC(\TT-r)}}~\OPT(\HH,w).
    \end{equation*}
\end{observation}

Given a search tree $T$ on a tree $\TT$ and a subgraph $\HH$ of $\TT$, we denote by $\LCA_T(\HH)$ the vertex with maximal depth in $T$ that is an ancestor of all vertices of $\HH$.

The following two technical lemmas will be useful later.

\begin{lemma}
\label{folklore-lemma}
    Let $T$ be a search tree on $\TT$ and $\HH$ a connected subgraph of $\TT$. Then $\LCA_T(\HH)\in V(\HH)$.
\end{lemma}

\begin{proof}
    Denote $\ell=\LCA_T(\HH)$ and assume $\ell\notin V(\HH)$. Denote $\tilde{\HH} = \TT[V(T_\ell)]$. We have $\HH\subseteq\tilde{\HH}$. Since $\HH$ is connected, between every two vertices of $\HH$, there is a path in $\tilde{\HH}$ that does not go through $\ell$. It follows that $\HH\subseteq \mathcal{C}$ for some connected component $\mathcal{C}\in\mathbb{C}(\tilde{\HH}-\ell)$. Then $\LCA_T(\mathcal{C})$ is a child of $\ell$ and an ancestor of all the vertices in $\HH$, a contradiction to the choice of $\ell$.
\end{proof}

\begin{lemma}\label{p:centroids-intersection}
	Let $\TT$ be a tree with weight function $w$, and let $\mathcal{A}$, $\mathcal{B}$ be connected subgraphs of $\TT$. If $\mathcal{A}$ and $\mathcal{B}$ each contain a centroid of $(\TT,w)$, and $V(\mathcal{A}) \cap V(\mathcal{B}) \neq \emptyset$, then $V(\mathcal{A}) \cap V(\mathcal{B})$ contains a centroid of $(\TT,w)$.
\end{lemma}
\begin{proof}
	Let $c_1 \in \mathcal{A}$ and $c_2 \in \mathcal{B}$ be centroids. If $c_1 \in V(\mathcal{A}) \cap V(\mathcal{B})$ or $c_2 \in V(\mathcal{A}) \cap V(\mathcal{B})$, we are done. Otherwise, there must be a vertex $c$ on the path between $c_1$ and $c_2$ such that $c \in V(\mathcal{A}) \cap V(\mathcal{B})$. We show that $c$ is a centroid.
	
	Let $\mathcal{M} \in \CC(\TT \setminus c)$. Note that either $c_1 \notin V(\mathcal{M})$ or $c_2 \notin V(\mathcal{M})$. Assume w.l.o.g., that $c_1 \notin V(\mathcal{M})$. Then $V(\mathcal{M}) \subseteq V(\mathcal{M'})$ for some component $\mathcal{M'} \in \CC(\TT \setminus c_1)$, and so $w(\mathcal{M}) \le w(\mathcal{M'}) \le \frac{1}{2}w(\TT)$. This concludes the proof.
\end{proof}

\paragraph{Projection of a search tree.}

For a rooted tree $T$ and a vertex $v\in V(T)$, we denote by $\Path_T(v)$ the set of vertices on the path in $T$ from $\root(T)$ to $v$, including both endpoints. The proofs of our upper bounds require the following notion of \emph{projection} of a search tree.

\begin{restatable}{theorem}{restateDefProjection}
\label{projection_thm}
    Let $T$ be a search tree on $\TT$ and $\HH$ a connected subgraph of $\TT$. There is a unique search tree $T|_\HH$ on $\HH$ such that for every $v\in V(\HH)$,
    \begin{equation*}
        \Path_{T|_\HH}(v) = \Path_T(v)\cap V(\HH).
    \end{equation*}
\end{restatable}

\begin{definition}[Projection]
\label{projection_def}
 Let $T$ be a search tree on $\TT$ and $\HH$ a connected subgraph of $\TT$. 
We call $T|_\HH$, whose existence is established by Theorem \ref{projection_thm}, the projection of $T$ to $\HH$.
\end{definition}

An equivalent notion of projection was used by Cardinal, Langerman and Perez-Lantero~\cite{Cardinal18} without explicitly stating \Cref{projection_thm}, and a closely related notion was used by Cicalese, Jacobs, Laber and Molinaro~\cite{Cicalese1} in the edge-query setting. In \Cref{appa} we prove \Cref{projection_thm}, present the definition of a projection used by Cardinal et al.\ and show that the two definitions are equivalent.

\paragraph{Median and centroid.}

A certain concept of a \emph{median vertex} of a tree has been used previously in the literature. If $\TT$ is a tree with positive vertex weights and positive edge weights, then the median of $\TT$ is the vertex $v$ minimizing the quantity $\sum_{u \neq v} w(u) \cdot d(u,v)$.
Here $w(u)$ is the weight of the vertex $u$, and $d(u,v)$ is the distance from $u$ to $v$, i.e., the sum of the edge weights on the path from $u$ to $v$. We show (proof in Appendix~\ref{appb2}) that if all edge-weights are $1$, then medians are precisely centroids.

\begin{restatable}{lemma}{lemrestatecm}\label{p:centroid-is-median}
	Let $\TT$ be a graph and $w$ be a weight function on $V(\TT)$. For each $u \in V(\TT)$, define $W(u) = \sum_{v\in V(\TT)}d_\TT(u,v)\cdot w(v)$, where $d_{\TT}(u,v)$ denotes the number of edges on the path from $u$ to $v$ in $\TT$. Then $c \in V(\TT)$ is a centroid of $(\TT, w)$ if and only if $W(c)$ is minimal.
\end{restatable}

\paragraph{Breaking ties.}

Our lower bounds require the following tie-breaking procedure, with proof in Appendix~\ref{appc}.

\begin{restatable}{lemma}{lemrestateq}
\label{tie_breaking_lemma}
    Let $\TT$ be a tree, $w:V(\TT)\rightarrow\mathbb{R}_{\geq0}$ a weight function and let $C$ be a centroid tree of $(\TT,w)$. For every $\varepsilon>0$ there exists a weight function $w':V(\TT)\rightarrow\mathbb{R}_{\geq0}$ such that $C$ is the unique centroid tree of $(\TT,w')$ and $\lVert w'-w\rVert_\infty<\varepsilon$.
\end{restatable}

\paragraph{Contractions in weighted trees.}

Let $\TT$ be a tree, let $w \colon V(\TT) \rightarrow \Rnn$ be a weight function, and let $\{u,v\}$ be an edge of $\TT$. Then \emph{contracting} $\{u,v\}$ produces a tree $\TT'$ and a weight function $w'$, defined as follows. The tree $\TT'$ is obtained from $\TT$ by removing $u$ and $v$, adding a new vertex $s$, and setting $N_{\TT'}(s) = N_\TT(u) \cup N_\TT(v) \setminus \{u,v\}$. The weight function $w'$ is defined as $w'(s) = w(u) + w(v)$ and $w'(x) = w(x)$ for $x \in V(\TT') \setminus \{s\}$.

The lemma below shows that contractions essentially preserve centroids.

\begin{lemma}\label{p:contr-preserves}
	Let $\TT$ be a tree and $w \colon V(\TT) \rightarrow \Rnn$ be a weight function, and let $\TT', w'$ be obtained from $\TT, w$ by contracting the edge $\{u,v\}$ into a new vertex $s$. Then $s$ is a centroid of $(\TT',w')$ if and only if $u$ or $v$ is a centroid of $(\TT,w)$; and $x$ is a centroid of $(\TT',w')$ if and only if $x$ is a centroid of $(\TT,w)$, for each $x \in V(\TT') \setminus \{s\}$.
\end{lemma}
\begin{proof}
	W.l.o.g., assume that $w(\TT) = 1$.
	Consider first a vertex $x \in V(\TT) \setminus \{u,v\}$. The contraction $\{u,v\}$ does not change the weight of any component $\HH \in \CC(\TT \setminus x)$. Thus, $x$ is a centroid of $(\TT,w)$ if and only if it is a centroid of $(\TT',w')$.
	
	Consider now the vertices $u$, $v$, and $s$. Observe that $\CC(\TT \setminus \{u,v\}) = \CC(\TT' \setminus s)$, and the contraction does not affect the weights of these components. If $u$ or $v$ is a centroid of $(\TT,w)$, then no component of $\CC(\TT \setminus \{u,v\})$ has weight more than $\frac{1}{2}$. Thus, $s$ is a centroid of $(\TT',w')$.
	
	On the other hand, suppose that $s$ is a centroid of $(\TT',w')$. Consider a component $\HH \in \CC(\TT \setminus u)$. If $\HH$ does not contain $v$, then $\HH \in \CC(\TT \setminus \{u,v\}) = \CC(\TT' \setminus s)$, so $w(\HH) \le \frac{1}{2}$ by assumption. In the same way, for each $\HH \in \CC(\TT \setminus v)$, if $\HH$ does not contain $u$, then $w(\HH) \le \frac{1}{2}$. Now let $\HH_v$ be the component of $\TT \setminus u$ that contains $v$, and let $\HH_u$ be the component of $\TT \setminus v$ that contains $u$.
	Since $\HH_v$ and $\HH_u$ are disjoint, we have $w(\HH_v) + w(\HH_u) \le 1$, so either $w(\HH_v) \le \frac{1}{2}$ or $w(\HH_u) \le \frac{1}{2}$. Thus, either $u$ or $v$ is a centroid of $(\TT',w')$.
\end{proof}

\section{Approximation guarantees for general trees}
\label{sec3}

In this section we prove the general upper and lower bounds on the approximation quality of centroid trees.

\subsection{Upper bound}

\restatethma*

We start with a lemma.

\begin{lemma}
\label{lemma1}
    Let $c$ be a centroid of $(\TT,w)$ and $m=w(\TT)$. Then
    \begin{equation}
    \label{eq1}
        \OPT(\TT,w) \geq \frac{m}{2} + \frac{w(c)}{2} + \sum_{\mathclap{\HH\in\CC(\TT-c)}}~\OPT(\HH,w).
    \end{equation}
\end{lemma}

\begin{proof}
    Let $T$ be an arbitrary search tree on $\TT$. We will show that $\cost_w(T)$ is at least the right hand side of \Cref{eq1}.
    
    Denote $r=\root(T)$. If $r=c$, using \Cref{observation99},  we have
    \begin{equation*}
        \cost_w(T) \geq m + \sum_{\mathclap{\HH\in\CC(\TT-c)}}\OPT(\HH,w),
    \end{equation*}
    which implies the claim. Assume therefore that $r\neq c$. Denote by $\HH^*$ the connected component of $\TT-c$ where $r$ is. The contribution of vertices of $\HH^*$ to $\cost_w(T)$ is at least $\cost_{w_{\HH^*}}(T|_{\HH^*})$. For $\HH\in\CC(\TT-c)$,  $\HH\neq \HH^*$ and $v\in V(\HH)$, we have $\Path_T(v)\supseteq\{r\}\cup\Path_{T|_\HH}(v)$, therefore the contribution of vertices of every $\HH\neq \HH^*$ is at least $w(\HH)+\cost_w(T|_\HH)$. Finally, the contribution of $c$ is at least $2w(c)$, since both $c,r\in\Path_T(c)$. Summing the contributions of all the vertices, we get
    \begin{align*}
        \cost_w(T) &\geq 2w(c) + \cost_w(T|_{\HH^*}) + \sum_{\HH\neq \HH^*}(w(\HH)+\cost_w(T|_\HH))\\
        &\geq m - w(\HH^*) + w(c) + \sum_\HH\OPT(\HH,w)\\
        &\geq \frac{m}{2} + w(c) + \sum_\HH\OPT(\HH,w),
    \end{align*}
    where the last inequality follows from $c$ being a centroid.
\end{proof}

\begin{proof}[Proof of \Cref{thm1}.]
    The proof is by induction on the number of vertices. When $|V(\TT)|=1$ we have
    \begin{equation*}
        2\cdot\OPT(\TT,w) - m = 2m - m = m = \CENT(\TT,w)
    \end{equation*}
    as required.

    Assume $|V(\TT)|>1$. Let $C$ be a centroid tree on $\TT$ and $c=\root(C)$. Using \cref{observation99} and the induction hypothesis we have
    \begin{align*}
        \cost_w(C) &~\leq~
        m + \sum_{\mathclap{ \HH\in\CC(\TT \setminus c)}}~\CENT(\HH,w)\\
        &~\leq~ m + \sum_{\mathclap{\HH\in\CC(\TT \setminus c)}}~\left(2\cdot\OPT(\HH,w) - w(\HH)\right)\\
           &~=~ w(c) + 2 \cdot \sum_{~\mathclap{\HH\in\CC(\TT \setminus c)}}~\OPT(\HH,w),
    \end{align*}
    therefore it is enough to show that
    \begin{equation*}
    \label{eq2}
        w(c) + 2 \cdot \sum_{\mathclap{\HH\in\CC(\TT \setminus c)}}~\OPT(\HH,w) ~\leq~ 2 \cdot \OPT(\TT,w) - m,
    \end{equation*}
    which is just a re-arrangement of \Cref{lemma1}. 
This concludes the proof. 
\end{proof}

We note that in the \emph{edge-query model} of search trees a $2$-approximation was shown in~\cite{Cicalese1} using similar techniques. In contrast to that result, however, our approximation guarantee is best possible.

\subsection{Lower bounds} 
In the following, we show the tightness of \Cref{thm1} through the following lower bounds.

\restatethmb*

\paragraph{Part (i).}
We proceed by constructing a sequence $(\TT_n,w_n)$ such that for \emph{some} centroid tree $C_n$,
\begin{equation*}
    \cost_{w_n}(C_n)\geq 2\cdot\OPT(\TT_n,w_n) - 1.
\end{equation*}
Using \Cref{tie_breaking_lemma}, we can then add an arbitrarily small perturbation to $w_n$ to make $C_n$ the \emph{unique} centroid tree. (Observe that for every search tree $T$, $\cost_w(T)$ is continuous in $w$, therefore so is $\OPT(\TT,w)$.)

The sequence $(\TT_n,w_n)$ is constructed recursively as follows. For the sake of the construction we view $\TT_n$ as a rooted tree. The base case $\TT_0$ is a tree with a single vertex $v$ and $w_0(v)=1$. For $n>0$, take two copies $(\A,w_\A)$ and $(\B,w_\B)$ of $(\TT_{n-1},w_{n-1})$. Connect the roots of $\A$ and $\B$ to a new vertex $c$. Finally, set $\root(\TT_n)=\root(\A)$ (see \Cref{fig:lower_bound_general}). We define $w_n$ as follows. (observe that $w_n(\TT_n) = 1$, by induction on $n$.)
\begin{equation*}
    w_n(v) = \begin{cases}
            0, & v=c\\
			\frac{1}{2}w_\A(v), & v\in V(\A)\\
			\frac{1}{2}w_\B(v), & v\in V(\B).
            \end{cases}
\end{equation*}

\begin{figure}[h]
\centering
\includegraphics[width=4in]{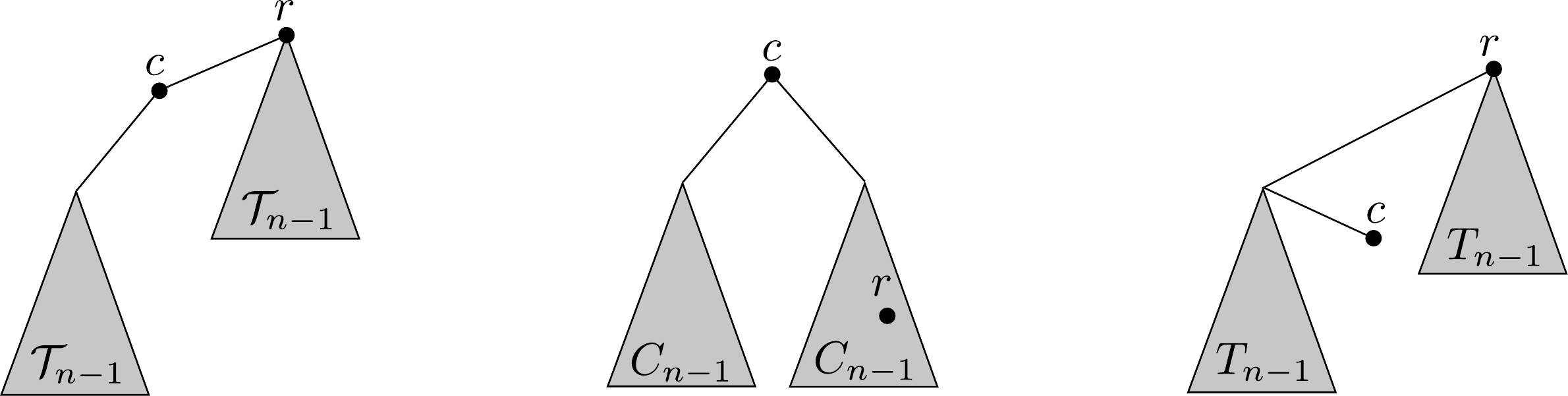}
\caption{Illustration of the proof of Theorem~\ref{thm2}(i). Vertex $r$ is the root of $\TT_n$. (\emph{Left.}) The underlying tree $\TT_n$. (\emph{Middle.}) The entroid tree $C_n$ of $\TT_n$. (\emph{Right.}) The search tree $T_n$.}
\label{fig:lower_bound_general}
\end{figure}

Let $C_n$ denote the search tree on $\TT_n$ obtained by setting $c$ as the root and recursing. Observe that $C_n$ is a centroid tree of $(\TT_n,w_n)$.

\begin{lemma}
    \label{lemma2}
    The following hold
    \begin{enumerate}[(a)]
        \item $\cost_{w_n}(C_n) = n + 1$,
        \item $\lim_{n\rightarrow\infty}\OPT(\TT_n,w_n)=\infty$.
    \end{enumerate}
\end{lemma}

\begin{proof}
    Let $c_n=\cost_{w_n}(C_n)$. Clearly, $c_0=1$. Assume $n>0$. Let $C_\A$ and $C_\B$ be search trees on $\A$ and $\B$ respectively, each a copy of $C_{n-1}$. By construction of $C_n$ we have
    \begin{equation*}
        c_n = 1 + \frac{1}{2}\cost_{w_\A}(C_\A) + \frac{1}{2}\cost_{w_\B}(C_\B) = 1 + c_{n-1},
    \end{equation*}
    and (a) follows by induction.
    
    Using (a) and \Cref{thm1}, part (b) follows:    \begin{equation*}
        {\OPT}(\TT_n,w_n) \geq \frac{c_n+1}{2} = \frac{n}{2} + 1 \rightarrow\infty. \qedhere
    \end{equation*}
\end{proof}
    
Next, in order to bound $\OPT(\TT_n,w_n)$ from above, we construct a sequence of search trees $T_n$ on $\TT_n$. For $n=0$, tree $T_0$ is a single vertex. Assume $n>0$. Let $\A$, $\B$, and $c$ be as in the definition of $\TT_n$. Let $T_\A$ and $T_\B$ be search trees over $\A$ and $\B$ respectively, each a copy of $T_{n-1}$. Denote $r_\A=\root(\A)$ and $r_\B=\root(\B)$. Tree $T_n$ is obtained by adding an edge from $r_\A$ to $r_\B$ and an edge from $r_\B$ to $c$, and setting $\root(T_n)=r_\A$.

\begin{lemma}
\label{lemma3}
    $\cost_{w_n}(T_n) = \frac{n}{2} + 1$.
\end{lemma}
    
\begin{proof}
    Denote $t_n = \cost_{w_n}(T_n)$. Clearly $t_0=1$. Assume $n>0$. The contribution of vertices of $\A$ to $t_n$ is exactly $\frac{1}{2}\cost_{w_\A}(T_\A)=\frac{t_{n-1}}{2}$. Since $r$ is an ancestor of all vertices in $\B$, the contribution of these vertices to $t_n$ is exactly $\frac{1}{2}(1+\cost_{w_\B}(T_\B))=\frac{1+t_{n-1}}{2}$. Summing the contribution of all vertices, we get $t_n=t_{n-1} + \frac{1}{2}$ and the claim follows.
\end{proof}

\begin{proof}[Proof of
\Cref{thm2}(i).]
    By \Cref{lemma3}, $\OPT(\TT_n,w_n)\leq \frac{n}{2}+1$. Together with \Cref{lemma2}, the claim follows.
\end{proof}

\medskip

\paragraph{Part (ii).}

Let $\TT_n$, $C_n$ and $T_n$ be as in the proof of \Cref{thm2}(i). Let $\mu_n:V(\TT_n)\rightarrow\Rnn$ be the constant function $\mu_n(v)=1$. Let $\A$, $\B$ and $c$ be as in the construction of $\TT_n$. Observe that $c$ is the unique centroid of $(\TT_n,\mu_n)$. By induction, $C_n$ is the unique centroid tree of $(\TT_n,\mu_n)$.

\begin{lemma}
\label{lemma4}
    The following hold
    \begin{enumerate}[(a)]
        \item $\cost_{\mu_n}(C_n) = 2^{n+1}n + 1$,
        \item $\cost_{\mu_n}(T_n) = 2^nn + 2^{n+1} - 1$.
    \end{enumerate}
\end{lemma}

\begin{proof}
    Denote $c_n=\cost_{\mu_n}(C_n)$ and $t_n=\cost_{\mu_n}(T_n)$. Denote also by $v_n$ the number of vertices $|V(\TT_n)|$. By induction, $v_n=2^{n+1}-1$. Observe that $c_n$ satisfies the recurrence
    \begin{equation*}
        c_n = v_n + 2c_{n-1} = 2^{n+1} - 1 + 2c_{n-1},
    \end{equation*}
    and that $c_n=2^{n+1}n + 1$ is the solution for this formula. Repeating the analysis from \Cref{lemma3}, $t_n$ satisfies the recurrence 
    \begin{equation*}
        t_n = t_{n-1} + (v_{n-1} + t_{n-1}) + 2 = 2^n + 1 + 2t_{n-1}.
    \end{equation*}
    Observe that $t_n=2^nn + 2^{n+1} - 1$ is the solution to this formula.
\end{proof}

\begin{proof}[Proof of \Cref{thm2}(ii)]
    Denote by $\tilde{\mu}_n$ the uniform distribution over $V(\TT)$. Using \Cref{lemma4} we have
    \begin{equation*}
        \frac{\cost_{\tilde{\mu}_n}(C_n)}{\OPT(\TT_n,\tilde{\mu}_n)} \geq \frac{\cost_{\tilde{\mu}_n}(C_n)}{\cost_{\tilde{\mu}_n}(T_n)} = 
        \frac{\cost_{\mu_n}(C_n)}{\cost_{\mu_n}(T_n)}
        \rightarrow 2.
    \end{equation*}
    We also have
    \begin{equation*}
        \cost_{\tilde{\mu}_n}(C_n) = \frac{\cost_{\mu_n}(C_n)}{2^{n+1}-1}\rightarrow\infty,
    \end{equation*}
    therefore, using \cref{thm1}, also $\OPT(\TT_n,\tilde{\mu}_n)\rightarrow\infty$. This concludes the proof.
\end{proof}

\section{Approximation guarantees for trees with bounded degrees}
\label{sec4}

In this section we show the upper and lower bounds on the approximation quality of centroid trees when the underlying tree $\TT$ has bounded degree. 

\subsection{Upper bound}

\restatethmd*

For simplicity, in what follows we omit the weight function $w$ from notations.

\begin{lemma}
\label{lemma5}
    Let $C$ be a centroid tree of $\TT$ such that $\cost(C)=\CENT(\TT)$. Let $P=(v_0,v_1,\dots,v_p)$ be any path in $C$. Then
    \begin{equation*}
        \CENT(\TT) \leq \left(2-\frac{1}{2^p}\right)m + \sum_{\mathclap{\HH\in\CC(\TT \setminus P)}}~\CENT(\HH).
    \end{equation*}
\end{lemma}

\begin{proof}
    By induction on $p$. For $p=0$ we have
    \begin{equation*}
        \CENT(\TT) = m + \sum_{\mathclap{\HH\in\CC(\TT \setminus v_0)}}~\CENT(\HH),
    \end{equation*}
    as required.
    
    Assume now $p>0$. Denote by $\tilde{\TT}$ the connected component of $\TT \setminus v_0$ where $v_1$ is. Denote $\tilde{P}=(v_1,\dots,v_p)$, and $\Tilde{m}=w(\Tilde{\TT})$. Observe that $\Tilde{m}\leq m/2$ and that $\CC(\TT \setminus P) = \CC(\Tilde{\TT} \setminus \Tilde{P})\cupdot(\CC(\TT-v_0)\setminus \{\Tilde{\TT}\})$. By the induction hypothesis we have
    \begin{equation*}
        \CENT(\tilde{\TT}) \leq \left(2-\frac{1}{2^{p-1}}\right)\Tilde{m} + \sum_{\mathclap{\HH\in\CC(\Tilde{\TT} \setminus \Tilde{P})}}~\CENT(\HH),
    \end{equation*}
    therefore
    \begin{align*}
        \CENT(\TT) &~=~ m + \CENT(\Tilde{\TT}) + \sum_{\mathclap{\substack{\HH\in\CC(\TT \setminus v_0)\\\HH\neq \Tilde{T}}}}~\CENT(\HH)\\
        & ~\leq~ m + \left(2-\frac{1}{2^{p-1}}\right)\Tilde{m} + \sum_{\mathclap{\HH\in\CC(\Tilde{\TT} \setminus \Tilde{P})}}~\CENT(\HH) + \sum_{\mathclap{\substack{\HH\in\CC(\TT \setminus v_0)\\H\neq \Tilde{T}}}}~\CENT(\HH)\\
        & ~\leq~ m + \left(2-\frac{1}{2^{p-1}}\right)\frac{m}{2} + \sum_{\mathclap{\HH\in\CC(\TT \setminus P)}}~\CENT(\HH)\\
        & ~=~ \left(2-\frac{1}{2^p}\right)m + \sum_{\mathclap{\HH\in\CC(\TT \setminus P)}}~\CENT(\HH),
    \end{align*}
    as required.
\end{proof}

\begin{proof}[Proof of \Cref{thm4}]
    The proof is by induction on $|V(\TT)|$. Let $T$ be any search tree on $\TT$. We will show that $\CENT(\TT)\leq \left(2-\frac{1}{2^\Delta}\right)\cost(T)$.
    
    Denote $r=\root(T)$. Let $C$ be a centroid tree on $\TT$ with $\cost(C)=\CENT(\TT)$. Denote by $v_0,v_1,\dots,v_d=r$ the vertices along the path to $r$ in $C$. Denote $\TT_i=\TT[V(C_{v_i})]$. Observe that $r\in V(\TT_d)\subseteq\dots\subseteq V(\TT_0)=V(\TT)$. For $i<d$, denote by $\mathcal{K}_i$ the connected component of $\TT \setminus r$ where $v_i$ is. Denote by $s_i$ the unique child of $r$ in $T$ such that $s_i\in V(\mathcal{K}_i)$, i.e., $V(T_{s_i}) = V(\mathcal{K}_i)$. Finally, denote by $p$ the minimal $i$ for which one of the following holds:
    \begin{enumerate}
        \item $v_i=r$, i.e., $i = d$,
        \item $s_i\in V(\TT_{i+1})$, or 
        \item there exists $j<i$ such that $\mathcal{K}_j=\mathcal{K}_i$, i.e., $s_j = s_i$.
    \end{enumerate}
    Note that from the third condition above it follows that $p\leq\Delta$. Denote $P=(v_0,\dots,v_p)$. We will prove the following.
    
    \begin{claim}
    \label{claim1}
        \begin{equation*}
            \cost(T)\geq m + \sum_{\mathclap{\HH\in\CC(\TT \setminus P)}}\cost(T|_\HH).
        \end{equation*}
    \end{claim}
    
    Assume for now that \Cref{claim1} holds. Using \Cref{lemma5}, the fact that $p\leq\Delta$, the induction hypothesis and \Cref{claim1}, we have
    \begin{align*}
        \CENT(\TT) &\leq \left(2-\frac{1}{2^p}\right)m + \sum_{\mathclap{\HH\in\CC(\TT-P)}}~\CENT(\HH)\\
        &\leq \left(2-\frac{1}{2^\Delta}\right)m + \sum_{\mathclap{\HH\in\CC(\TT \setminus P)}}~~\left(2-\frac{1}{2^\Delta}\right)\cost(T|_\HH)\\
        &\leq \left(2-\frac{1}{2^\Delta}\right)\cost(T) ,
    \end{align*}
    as required. It remains to prove \Cref{claim1}.
\end{proof}

\begin{proof}[Proof of \Cref{claim1}]
    The proof breaks into three cases according to the defining condition of $p$.
    
    \paragraph{Case 1.} Assume $v_p=r$. For every $\HH\in\CC(\TT \setminus P)$ and $v\in V(\HH)$ we have $\Path_T(v)\supseteq\{r\}\cupdot \Path_{T|_\HH}(v)$. The contribution of such $v$ to $\cost(T)$ is therefore at least $w(v)(1+|\Path_{T|_\HH}(v)|)$. The contribution of each $v_i$ to $\cost(T)$ is at least $w(v_i)$. Summing the contribution of all vertices yields the required result. See \Cref{fig:cases12}.
    
    \begin{figure}
        \centering
        \includegraphics[width=\textwidth]{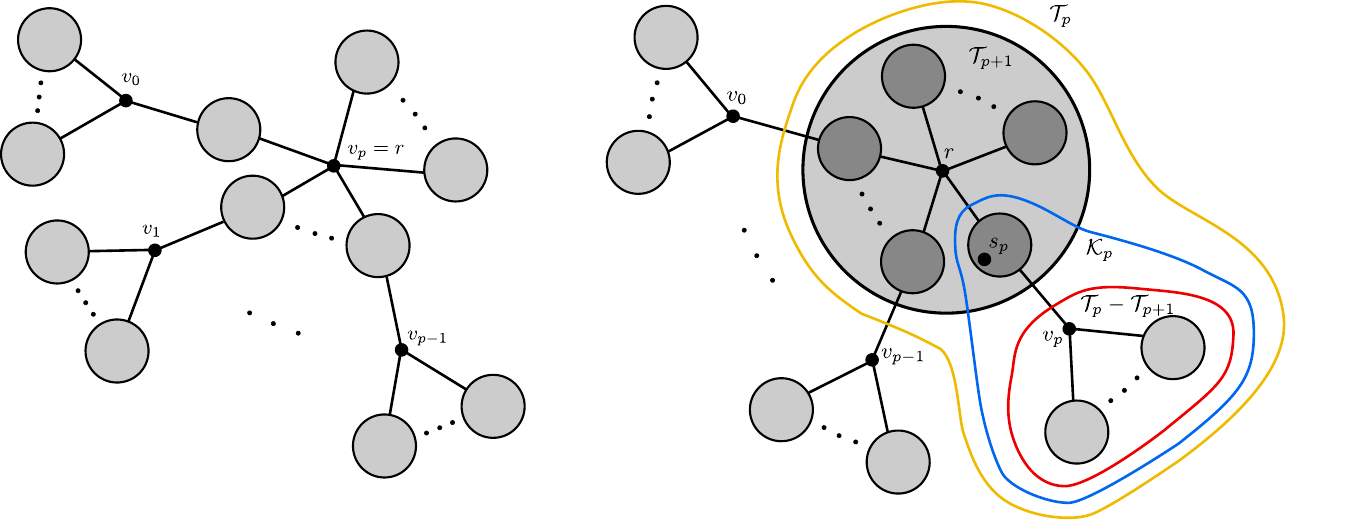}
        \caption{Illustration of the proof of \Cref{claim1}. Connected components of $\TT-P$ are represented by light gray circles. (\emph{Left.}) Case 1. (\emph{Right.}) Case 2. Vertices in $\TT-\TT_p$ have $r$ as ancestor. Vertices in $\TT_p-\TT_{p+1}$ have both $r$ and $s_p$ as ancestors.}
        \label{fig:cases12}
    \end{figure}

    \paragraph{Case 2.} Assume $s_p\in V(\TT_{p+1})$. Denote by $\CC_1$ the set of connected components of $\TT \setminus P$ that are not contained in $\TT_p$. Denote $\CC_2 = \CC(\TT \setminus P)\setminus\CC_1$ (see \Cref{fig:cases12}). For every $\HH\in\CC_1$, if $v\in V(\HH)$, then $\Path_T(v)\supseteq\{r\}\cupdot \Path_{T|_\HH}(v)$. Therefore the contribution of vertices in $V(\TT\setminus \TT_p)$ to $\cost(T)$ is at least
    \begin{equation}
    \label{eq3}
        m - w(\TT_p) + \sum_{\HH\in\CC_1}\cost(T|_\HH).
    \end{equation}
    For vertices $v\in V(\TT_p\setminus\TT_{p+1})$ we have $\{r, s_p\}\subseteq \Path_{T}(v)$. Therefore, using the fact that $w(\TT_p\setminus \TT_{p+1})\geq \frac{w(\TT_p)}{2}$, the contribution of vertices in $V(\TT_p)$ to $\cost(T)$ is at least
    \begin{equation}
    \label{eq4}
        2\cdot w(\TT_p\setminus \TT_{p+1}) + \sum_{\HH\in\CC_2}\cost(T|_\HH) \geq w(\TT_p) + \sum_{\HH\in\CC_2}\cost(T|_\HH).
    \end{equation}
    Summing \Cref{eq3} and \Cref{eq4} yields the required result.
    
    \paragraph{Case 3.} Assume that there is a $j<p$ such that $\mathcal{K}_p=\mathcal{K}_j$. Since $p$ is minimal, we further assume that case 2 did not occur for indices smaller that $p$. In particular, $s_p=s_j\notin V(\TT_p)$. As in case 2, the contribution of vertices in $\TT\setminus \TT_p$ to $\cost(T)$ is at least as in \Cref{eq3}. We have $r\notin V(\TT_p-\TT_{p+1})$ and $v_p\in V(\TT_p-\TT_{p+1})\cap V(\mathcal{K}_p)\neq\emptyset$, therefore, since $\TT_p-\TT_{p+1}$ is connected, $V(\TT_p-\TT_{p+1})\subseteq V(\mathcal{K}_p)$ (see \Cref{fig:case3}). It follows that vertices in $\TT_p-\TT_{p+1}$ have both $r$ and $s_p$ as ancestors. Since $w(\TT_p\setminus \TT_{p+1})\geq \frac{w(\TT_p)}{2}$, the contribution of vertices in $\TT_p$ is at least as in \Cref{eq4}. As in Case 2, the result follows by summing \Cref{eq3} and \Cref{eq4}.
\end{proof}

\begin{figure}
    \centering
    \includegraphics[width=\textwidth]{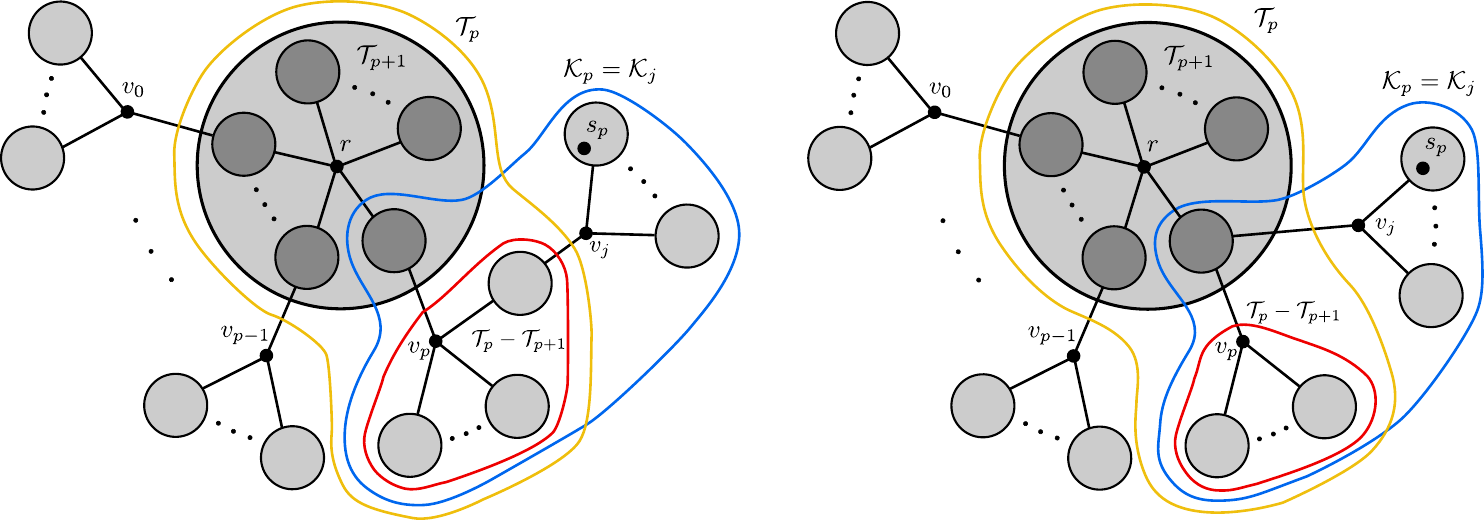}
    \caption{Case 3 in the proof of \Cref{claim1}. (\emph{Left.}) Sub-case where $v_p$ is in the path in $\TT$ between $r$ and $v_j$. (\emph{Right.}) Complementary sub-case. Connected components of $\TT-P$ are represented by light gray circles. In both sub-cases, vertices in $\TT_p\setminus\TT_{p+1}$ have both $r$ and $s_p$ as ancestors.}
    \label{fig:case3}
\end{figure}

\subsection{Lower bounds}
\restatethme*

\begin{figure}[h]
    \centering
    \includegraphics[width=\textwidth]{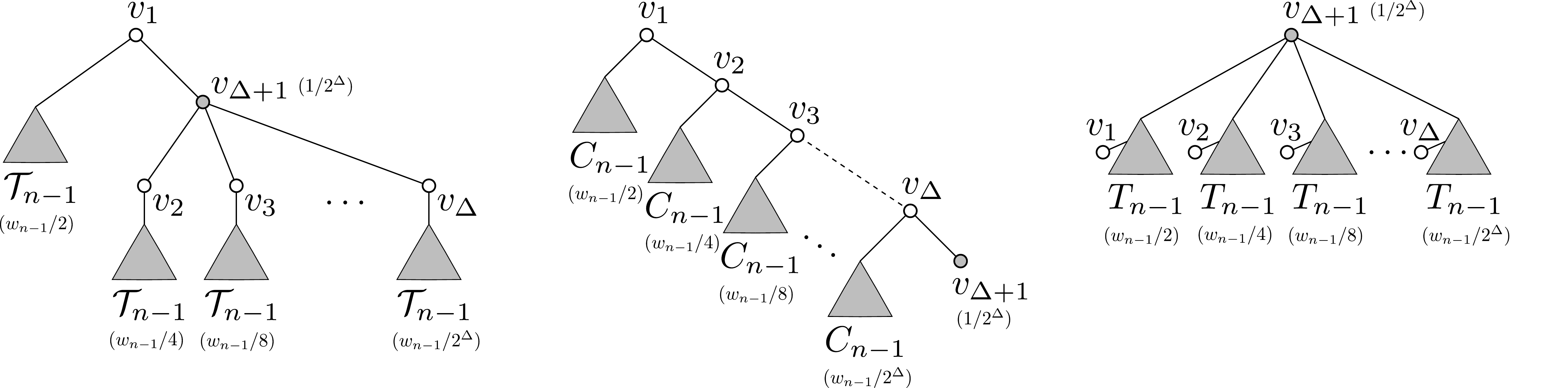}
    \caption{Illustration of \cref{thm5}(i). (\emph{Left.}) The underlying tree $\TT_n$. (\emph{Middle.}) The centroid tree $C_n$. (\emph{Right.}) The search tree $T_n$.}
    \label{fig:degree_lb}
\end{figure}

\paragraph{Part (i).}
As in the proof of \Cref{thm2}, it will suffice to prove \Cref{thm5}(i) for \emph{some} centroid tree $C_n$. Using \Cref{tie_breaking_lemma}, we can then add arbitrarily small perturbation to $w_n$, making $C_n$ the unique centroid tree.

The sequence $(\TT_n,w_n)$ of \Cref{thm5}(i) is constructed recursively as follows. For the sake of the construction we regard $\TT_n$ as a rooted tree. $\TT_0$ is simply a single vertex $v$ (which is the root) and $w_0(v)=1$. For $n>0$, $\TT_n$ is constructed from $\Delta$ copies of $\TT_{n-1}$ and $\Delta+1$ additional vertices, $v_1,\dots,v_{\Delta+1}$, as shown in \cref{fig:degree_lb} (left). The $i$'th copy of $\TT_{n-1}$ gets the weight function $w_{n-1}/2^i$. We set $w_n(v_i)=0$ for $1\leq i\leq\Delta$ and $w_n(v_{\Delta+1})=1/2^\Delta$. Finally, we set $\mathrm{root}(\TT_n)=v_1$. By induction, $\TT_n$ has maximal degree $\Delta$ and $w_n$ is a distribution on $V(\TT_n)$.

Let $C_n$ be a search tree on $\TT_n$ defined recursively as follows. Connect the vertices $v_1,\dots,v_{\Delta+1}$ to form a path and set $v_1$ as the root of $C_n$. Continue recursively on each connected component of $\TT\setminus\{v_1,\dots,v_{\Delta+1}\}$. (See \cref{fig:degree_lb} (middle).) Observe that $C_n$ is a centroid tree of $(\TT_n,w_n)$. 

\begin{lemma}
\label{lemma6}
For all $n$,
\begin{equation}
\label{eq5}
    \cost_{w_n}(C_n) = 2^{\Delta + 1} - 1 - (2^{\Delta + 1} - 2)
    \left(1-\frac{1}{2^\Delta}\right)^n.
\end{equation}
\end{lemma}

\begin{proof}
Denote $c_n=\cost_{w_n}(C_n)$. Clearly $c_0=1$ as required. Let $n>0$. For each $i$, the subtree of all the descendants of $v_i$ in $C_n$ has weight $1/2^{i-1}$. Therefore
\begin{equation*}
    c_n = \sum_{i=1}^{\Delta+1}\frac{1}{2^{i-1}} + \sum_{i=1}^{\Delta}\frac{1}{2^i}c_{n-1} = 2 - \frac{1}{2^\Delta} + \left(1 - \frac{1}{2^\Delta}\right)c_{n-1}.
\end{equation*}
It is straightforward to verify that the right hand side of \Cref{eq5} is the solution to the recursive formula above.
\end{proof}

In order to upper bound $\OPT(\TT_n,w_n)$ we construct recursively a search tree $T_n$ on $\TT_n$. For $n>0$, $T_n$ is constructed by setting $v_{\Delta+1}$ as root and attaching to it $\Delta$ copies of $T_{n-1}$. The vertices $v_1,\dots,v_\Delta$ are finally attached as leaves of $T_n$, each at it's unique valid place. See Figure~\ref{fig:degree_lb} (right).

\begin{lemma}
\label{lemma7}
For all $n$,
\begin{equation}
\label{eq6}
    \cost_{w_n}(T_n) = 2^\Delta - 2^\Delta\left(1-\frac{1}{2^\Delta}\right)^{n+1}.
\end{equation}
\end{lemma}

\begin{proof}
Denote $t_n=\cost_{w_n}(T_n)$. We have $t_0=1$. For $n>0$, $t_n$ obeys the recursive relation
\begin{equation*}
    t_n = 1 + \sum_{i=1}^\Delta\frac{1}{2^{i}}t_{n-1} = 1 + \left(1-\frac{1}{2^\Delta}\right)t_{n-1},
\end{equation*}
of which the right hand side of \Cref{eq6} is the solution.
\end{proof}

\begin{proof}[Proof of \Cref{thm5}(i)]
    Using \Cref{lemma6} and \Cref{lemma7},
    \begin{equation*}
        \frac{\cost_{w_n}(C_n)}{\OPT(\TT_n,w_n)}
        \geq
        \frac{\cost_{w_n}(C_n)}{\cost_{w_n}(T_n)}
        \rightarrow
        2 - \frac{1}{2^\Delta}. \qedhere
    \end{equation*}
\end{proof}

Observe that, as discussed in Section~\ref{sec1}, in the construction of \Cref{thm5}(i), $\OPT(\TT_n,w_n)/w_n(\TT_n)$ is bounded.

\medskip
\paragraph{Part (ii).}

To prove \Cref{thm5}(ii), we repeat the recursive construction of \Cref{thm5}(i) with a slight modification. As before, $\TT_0$ is a tree with a single vertex. For $n>0$, $(\TT_n,w_n)$ is constructed from $\Delta$ weighted copies of $(\TT_{n-1},w_{n-1})$ and $\Delta-1$ additional vertices, $v_1,\dots,v_{\Delta-1}$, each with weight $0$, as shown in \Cref{fig:lower_bound_bounded_deg} (left). We set $\root(\TT_n)=v_1$.

As before, the search tree $C_n$ is defined by connecting the vertices $v_1,\dots,v_{\Delta-1}$ to a path, setting $v_1$ as root and recursing on the remaining connected component. Observe that $C_n$ is a centroid tree of $(\TT_n,w_n)$. (See \cref{fig:lower_bound_bounded_deg} (middle).) The search tree $T_n$ is defined by setting $v_{\Delta-1}$ as root, attaching to it $\Delta$ copies of $T_{n-1}$, then adding the vertices $v_1,\dots,v_{\Delta-2}$ as leaves, each at its unique valid place. See \cref{fig:lower_bound_bounded_deg} (right).

\begin{figure}[h]
\centering
\includegraphics[width=\textwidth]{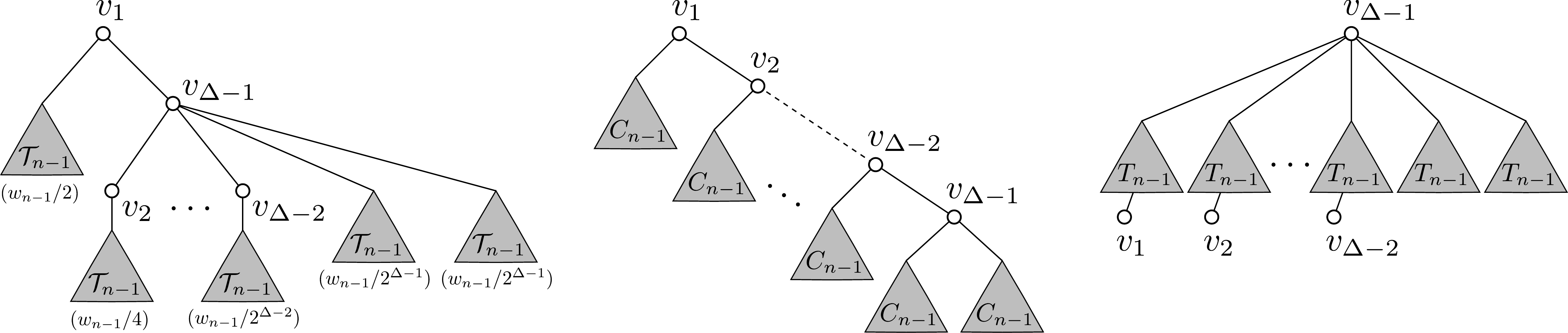}
\caption{An illustration of \cref{thm5}(ii). (\emph{Left.}) The underlying tree $\TT_n$. (\emph{Middle.}) The centroid tree $C_n$. (\emph{Right.}) The search tree $T_n$.}
\label{fig:lower_bound_bounded_deg}
\end{figure}

\begin{lemma}
\label{lemma8}
For all $n$,
\begin{enumerate}[(a)]
    \item $\cost_{w_n}(C_n) = \left(2-\frac{4}{2^\Delta}\right)\cdot n + 1$,
    \item $\cost_{w_n}(T_n) = n + 1$.
\end{enumerate}
\end{lemma}

The proof follows an analysis similar to that of \Cref{lemma6} and \Cref{lemma7}.

\begin{proof}
Denote $c_n=\cost_{w_n}(C_n)$ and $t_n=\cost_{w_n}(T_n)$. Clearly $c_0=t_0=1$. For $n>0$ we have
\begin{equation*}
    c_n = \sum_{i=1}^{\Delta-1}\frac{1}{2^{i-1}} + \sum_{i=1}^{\Delta-2}\frac{1}{2^i}c_{n-1} + 2\frac{1}{2^{\Delta-1}}c_{n-1} = 2-\frac{4}{2^\Delta} + c_{n-1}
\end{equation*}
and
\begin{equation*}
    t_n = \sum_{i=1}^{\Delta-2}\frac{1}{2^{i}}(t_{n-1} + 1) + 2\frac{1}{2^{\Delta-1}}(t_{n-1} + 1) = 1 + t_{n-1},
\end{equation*}
and the lemma follows by induction.
\end{proof}

\begin{proof}[Proof of \Cref{thm5}(ii)]
    The fact that $\lim_{n\rightarrow\infty}\OPT(\TT_n,w_n)=\infty$ follows from \Cref{lemma8} and \Cref{thm4} (or \Cref{thm1}). Using \Cref{lemma8} again, we have
    \begin{equation*}
        \cost_{w_n}(C_n) \geq \left(2-\frac{4}{2^\Delta}\right)\OPT(\TT_n,w_n) - 1+\frac{4}{2^\Delta}.
    \end{equation*}
    Using \Cref{tie_breaking_lemma}, for each $n$ we can add small enough perturbation to $w_n$ such that $C_n$ is the unique centroid tree and the claimed bound holds.
\end{proof}

\section{Computing centroid trees}
\label{sec5}

In this section, we show how to compute centroid trees using the \emph{top tree} framework of Alstrup, Holm, de Lichtenberg, and Thorup~\cite{AlstrupEtAl2005}. \emph{Top trees} are a data structure used to maintain dynamic forests under insertion and deletion of edges.
Most importantly, they expose a simple interface that allows the user to maintain information in the trees of the forest. For this, the user only needs to implement a small number of internal operations.

Alstrup et al.\ in particular show how to maintain the \emph{median} of trees in $\fO( \log n )$ per operation, see Section~\ref{sec2} for the definition of the median. As mentioned before, if all edge-weights are $1$, then medians are precisely centroids (see \cref{p:centroid-is-median}).

\begin{theorem}[{\cite[Theorem 3.6]{AlstrupEtAl2005}}]\label{thm-top-tree-centroid}
	We can maintain a forest with positive vertex weights on $n$ vertices under the following operations:
	\begin{itemize}
		\itemsep0pt
		\item Add an edge between two given vertices $u,v$ that are not in the same connected component;
		\item Remove an existing edge;
		\item Change the weight of a vertex;
		\item Retrieve a pointer to the tree containing a given vertex;
		\item Find the centroid of a given tree in the forest.
	\end{itemize}
	Each operation requires $\fO( \log n )$ time. A forest without edges and with $n$ arbitrarily weighted vertices can be initialized in $\fO(n)$ time.
\end{theorem}

Note that \cref{thm-top-tree-centroid} only admits \emph{positive} vertex weights, whereas we allowed zero-weight vertices. We show how to handle this problem in \cref{sec:spread-zero}.

We now show how to use \cref{thm-top-tree-centroid} to construct a centroid tree in $\fO( n \log n )$ time.

\begin{theorem}
	Given a tree $\TT$ on $n$ vertices and a positive weight function $w$, we can compute a centroid tree of $(\TT,w)$ in $\fO(n \log n)$ time.
\end{theorem}
\begin{proof}
	First build a top tree on $\TT$ by adding the edges one-by-one, in $\fO(n \log n)$ time. Then, find the centroid $c$, and remove each incident edge. Then, recurse on each newly created tree (except for the one containing only~$c$).
	The algorithm finds each vertex precisely once and removes each edge precisely once, for a total running time of $\fO(n \log n)$.
\end{proof}

\subsection{Output-sensitive algorithm}

We now improve the algorithm given above to run in time $\fO( n \log h )$, where $n$ is the number of vertices in $\TT$ and $h$ is the height of the computed centroid tree.

The main idea of the algorithm is inspired by the linear-time algorithm for \emph{unweighted} centroids by  Della Giustina, Prezza, and Venturini~\cite{Giustina}. Instead of building a top tree on the whole tree $\TT$, we first split $\TT$ into connected subgraphs of size roughly $h$, and build a top tree on each component. Contracting each component into a single vertex yields \emph{super-vertices} in a \emph{super-tree}.
Each search for a centroid consists of a global search and a local search: We first find the super-vertex containing the centroid, then we find the centroid within that super-vertex.
After finding the centroid, we remove it, which may split up the super-vertex into multiple super-vertices with a top tree each, and also may split the super-tree into a super-forest. Finally, we recurse on each component of the super-forest.

It can be seen that the total number of top tree operations needed is $\fO(n)$. Since the top trees each contain only $h$ vertices, a top tree operation takes $\fO(\log h)$ time, for a total of $\fO(n \log h)$. We now proceed with a more detailed description of the algorithm.

\paragraph{Ternarization.} If the degree of a vertex of $\TT$ is unbounded, then a partition into similarly sized connected subgraphs may not be possible. To fix this, we \emph{ternarize} $\TT$ by replacing each vertex $v$
of degree $d(v)$ larger than three by a path $P_v$ of $d(v)-2$ vertices with degree three.
Call the new vertices \emph{virtual} and let $\TT'$ denote the resulting tree. 
Each virtual vertex of $P_v$ is incident in $\TT'$ to an (arbitrary) edge incident to $v$ in $\TT$ except for the two endpoints of $P_v$ that are incident to two such edges each.
We maintain a link between each vertex in $\TT$ and every associated virtual vertex in $\TT'$ (if any).

Let $w'$ be a weight function on $\TT'$ obtained from $w$ by arbitrarily distributing weight from each deleted vertex to its associated virtual vertices. Note that $\TT$ and $w$ can be obtained from $\TT'$ and $w'$, respectively, by contracting every group of associated virtual vertices. Further bserve that $|V(\TT')| \le 2n$.

\paragraph{Partition.} Fix a parameter $k$. We now compute a partition into $\fO(\frac{n}{k})$ connected subgraphs of size at most $3k$ as follows. Arbitrarily root $\TT'$. Iteratively remove minimal rooted subtrees of size at least $k$, using a simple linear-time bottom-up traversal. Since each node has at most three children, this produces connected subgraphs of size between $k$ and $3k$. The only exception are the nodes remaining at the end, which we put into a possibly smaller subgraph. The total number of subgraphs is at most $\frac{2n}{k} + 1$.

\paragraph{Building the super-tree.} By contracting each connected subgraph of the partition into a single vertex, we obtain a tree $\stree$, the \emph{super-tree}, and a weight function $W$. We call each vertex $A \in V(\stree)$ a \emph{super-vertex}, and write $\TT'[A]$ for the subgraph of $\TT'$ contracted into $A$. We write $V(A) = V(\TT'[A])$ for short. By definition of $W$ via the contraction, we have $W(A) = w'(V(A))$.

Note that each \emph{super-edge} in $E(\stree)$ is associated with precisely one normal edge in $\TT'$; we maintain an explicit link between the super-edge and the normal edge. For each super-vertex $A$, we build a top tree on $\TT'[A]$, and store the weight $W(A)$. We call a vertex $v \in V(A)$ that is adjacent to some vertex $u \in \TT' \setminus V(A)$ a \emph{boundary vertex} of $A$, and maintain a list of boundary vertices for each super-vertex.

Constructing $\stree$ and $W$ can be done in linear time. Setting up the top trees requires $\fO(\frac{n}{k} \cdot k \log k) = \fO(n \log k)$ time.

\paragraph{Main procedure and recursion.}
Below, we describe how to find a centroid and remove it, along with associated virtual vertices. Doing so may split up $\stree$ (and implicitly $\TT$ and $\TT'$) into multiple connected components, on which we recurse.
Hence, a recursive step operates on a tree $\TT_\rr$ (a subgraph of $\TT$), a ternarization $\TT_\rr'$ of $\TT_\rr$, and a super-tree $\stree_\rr$ on $\TT_\rr'$. Note that only $\stree_\rr$ is explicitly given, whereas $\TT_\rr$ and $\TT_\rr'$ are implicit in the super-tree data structure. Our task is to find a centroid $c$ of $\TT_\rr$ and remove it from $\TT_\rr$, i.e., for each component $\HH$ of $\TT_\rr \setminus c$, we return a super-tree on a ternarization of $\HH$.

\paragraph{Finding centroids.} We now describe how to find a centroid of $\TT_\rr'$ using the super-tree $\stree_\rr$. Note that, by \cref{p:contr-preserves}, this is enough to find a centroid of $\TT_\rr$.

First, we find the centroid $A^*$ of $\stree_\rr$ with the trivial linear scan described in the introduction.

By \cref{p:contr-preserves}, a centroid of $\TT_\rr'$ must be contained in $\TT_\rr'[A^*]$.
We now construct a suitable weight function $w^*$ on $V(A^*)$ so that we can find a centroid of $\TT_\rr'$ within $\TT_\rr'[A^*]$. For each $v \in V(A^*)$, let $C_v$ be the connected component of $\TT_\rr' \setminus (A^* \setminus v)$ that contains $v$. Let $w^*(v) = w'(C_v)$. Note that $\TT_\rr'[A^*]$ and $w^*$ correspond to the tree and weight function obtained by contracting each $C_v$ into a single vertex. If $C_v = \{v\}$, i.e., $v$ is not a boundary vertex, then there are no contractions and $w^*(v) = w'(v)$.

We compute $w'(C_v)$ for each boundary vertex $v$, and temporarily modify the weight of each $v$ in the top tree on $V(A^*)$ to match $w^*(v)$. Then, we find a centroid $c$ of $\TT_\rr'[A^*]$ w.r.t.\ $w^*$, and undo the weight change.\footnote{Undoing the weight change is necessary, since we re-use the top tree data structures in later steps.}

By \cref{p:contr-preserves}, a centroid of $\TT_\rr'$ must be contained in $C_c$. We also know a centroid must be contained in $A^*$. Thus, by \cref{p:centroids-intersection}, a centroid must be contained in $C_c \cap A^* = \{c\}$, so $c$ must be a centroid of $(\TT_\rr',w')$. If $c$ is not virtual, it also is a centroid of $(\TT_\rr,w)$. If it is virtual, the linked non-virtual vertex in $\TT_\rr$ is a centroid of $(\TT_\rr,w)$.

The running time to find $A^*$ is $\fO(|V(\stree_\rr)|)$. Computing $w(C_v)$ for each boundary vertex $v$ can be done while traversing $\stree_\rr$, also in $\fO(|V(\stree_\rr)|)$ time. 
There are at most $\deg_{\stree_\rr}(A^*) \le |V(\stree_\rr)|$ boundary vertices, so changing the weights in the top tree takes $\fO(|V(\stree_\rr)| \log k)$ time.

\begin{figure}
	\centering
	\includegraphics[scale=0.85]{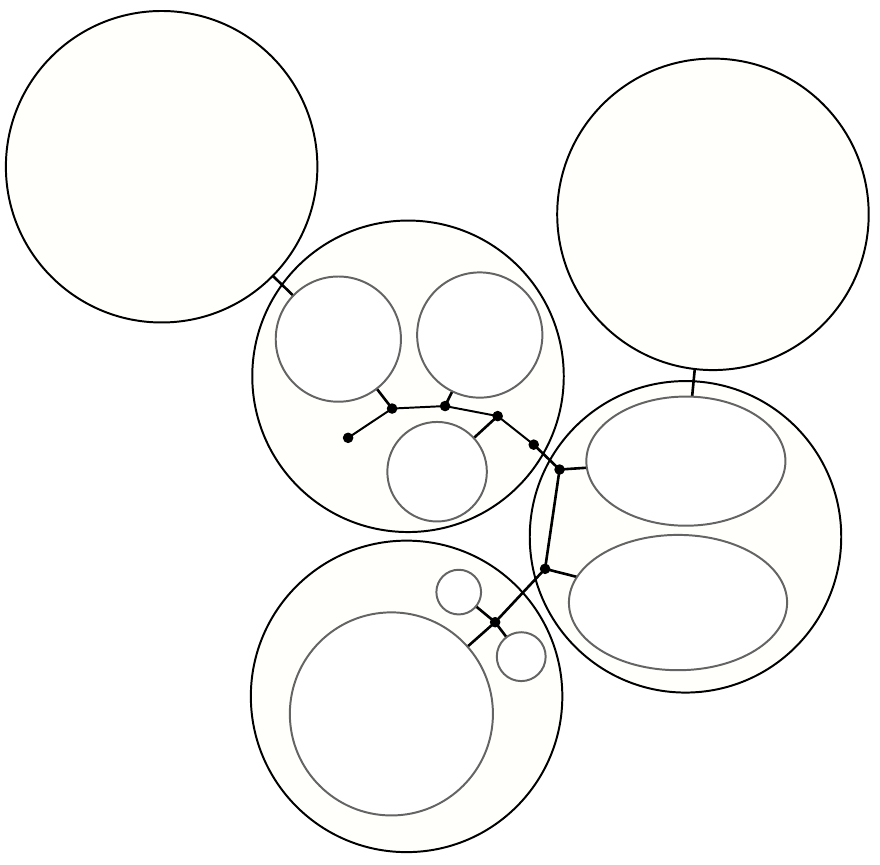}
	\hspace{10mm}
	\includegraphics[scale=0.85]{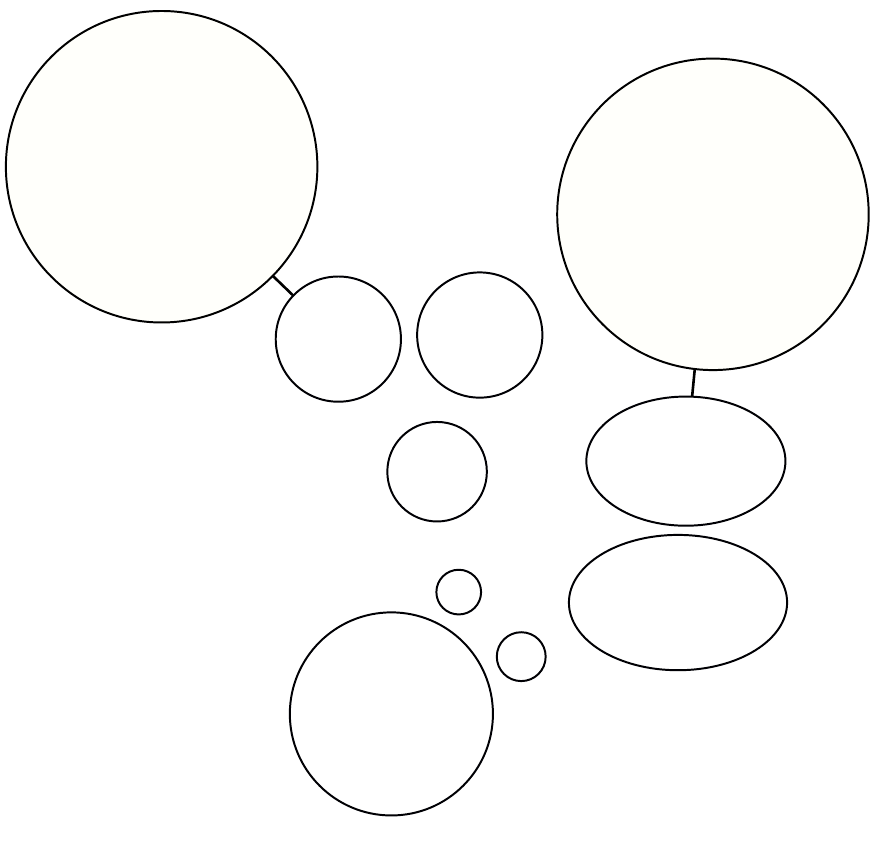}
	\caption{A splitting step. (\emph{Left.}) The super-tree before splitting. The large black circles are super-vertices. The small dots are is the set of vertices $R$ to be removed. The gray circles and ellipses are components obtained after removing $R$. (\emph{Right.}) The ten super-vertices within eight super-trees obtained after splitting.}\label{fig:splitting}
\end{figure}

\paragraph{Splitting the data structure.} We now remove $c$ and its virtual vertices from $\TT_\rr'$ and $\stree_\rr$. Let $R$ be the set of vertices to be removed. First, we remove each $v \in R$ from its associated top tree by deleting all of its incident edges. This may split up each of the top trees into multiple new top trees. Each new top tree corresponds to a new super-vertex.

The creation of new super-vertices changes the super-tree, and may even split it up into a super-forest (see \cref{fig:splitting} for an example). Let $A$ be a super-vertex from which we removed a vertex. For each boundary vertex $u$ of $A$, we find the new top tree that $u$ belongs to. Using this, we can compute all super-edges incident to the new super-vertices (recall that each super-edge corresponds to a normal edge between two boundary vertices in different super-vertices). Finally, we compute the connected components of the new super-forest $\sforest_\rr$ using a simple traversal.

Removing a vertex $v$ from the top trees requires $\fO(\deg_{\TT_\rr'}(v) \log k)$ time. Since each boundary vertex is the endpoint of an edge in $\stree_\rr$, there are at most $\fO( |E(\stree_\rr)| )$ boundary vertices. The running time of recomputing the super-forest is thus $\fO( |E(\stree_\rr)| \log k + |V(\sforest_\rr)| )$.

\paragraph{Recursion.} After splitting, we recurse on each component of the super-forest that contains more than one (normal) vertex.

\paragraph{Total running time.} The preprocessing time is $\fO(n \log k)$, as desired.

Consider one search-and-split step. Let $\TT_\rr'$ be the input tree, let $\stree_\rr$ be the input super-tree, let $R$ be the set of vertices removed in the splitting step, and let $\sforest_\rr$ be the super-forest produced after splitting. The time required to execute the step is
\begin{align*}
	\fO\left((|V(\stree_\rr)|+|E(\stree_\rr)|) \log k + |V(\sforest_\rr)| + \sum_{v \in R} \deg_{\TT_\rr'}(v) \log k\right).
\end{align*}

Since each vertex is removed only once, the third term sums up to $\fO(n \log k)$ over all vertices. The second term can be charged to recursive calls (for each component of $\sforest_\rr$ with more than one normal vertex) or single vertices (for each component of $\sforest_\rr$ consists of only one normal vertex). It remains to bound the first term.

If $\stree_\rr$ consists of only one super-vertex $A$, then we charge the $\fO( \log k)$ cost to the centroid of $\TT_\rr'[A]$, for a total of $\fO(n \log k)$ over the course of the algorithm.

We analyze the other recursive calls in rounds. Assume that in each round, we execute one step on each remaining super-tree, thereby finding all centroids on a certain level of the centroid tree. Let $\sforest_0$ be the initial super-tree, and let $\sforest_i$ be the forest of super-trees after round $i$. Note that the number of \emph{edges} cannot grow by splitting, so each $\sforest_i$ contains at most $\fO(\frac{n}{k})$ edges, and therefore at most $\fO(\frac{n}{k})$ non-isolated vertices, so the round requires $\fO( \frac{n}{k} \log k)$ time in total.

If the height of the tree is $h$, then we have precisely $h$ rounds. Thus, the running time of the algorithm is $\fO( (\frac{h}{k} + 1) n \log k)$.

In particular, if we know $h$ and set $k = h$, then the running time is $\fO( n \log h )$. If we do not know $h$, we start with $k = 2$ and run the algorithm for $k$ rounds. If it stops, then $h \le k$ and we are done. Otherwise, try again with $k \gets k^2$. The last run of the algorithm, where $h \le k \le h^2$, dominates the running time with $\fO( n \log h)$. Thus, we have

\begin{theorem}\label{p:construct-general}
	Let $\TT$ be a tree on $n$ vertices and $w$ be a positive weight function. We can compute a centroid tree of $(\TT,w)$ in time $\fO( n \log h )$, where $h$ is the height of the computed centroid tree.
\end{theorem}

\subsection{Spread and zero-weight vertices}\label{sec:spread-zero}

Let $w$ be a weight function on a tree $\TT$ such that at least one vertex weight is positive.
The \emph{spread} of $w$ is defined as $\sigma_w = w(\TT) / \min_{v \in V(\TT), w(v) > 0} w(v)$.

Let $T$ be a centroid tree of $(\TT,w)$. By definition, if $u$ is a parent of $v$, then $w(T_v) \le \frac{1}{2}w(T_u)$. More generally, if $v$ is at depth $d$, then $w(T_v) \le 2^{-d} w(\TT)$. Thus, the depth of a positive-weight vertex $v$ cannot be more than $1 + \log \frac{w(\TT)}{w(v)} \le 1 + \log \sigma_w$.

In particular, if all vertex weights are positive, then the height of every centroid tree is at most $1 + \log \sigma_w$, so the algorithm of the previous section runs in time $\fO(n \log \log \sigma)$.

If we allow zero-weight vertices, however, then this is not necessarily true, since there could be a large unbalanced zero-weight subtree. Also, the top-tree component of our algorithm does not allow zero weights for technical reasons. We believe that the top tree implementation of~\cite{AlstrupEtAl2005} can be adapted to this changed requirement, but for a simpler presentation, we prefer to use the standard interface, as defined in~\cite{AlstrupEtAl2005}.

The solution to both problems is to change the zero weights by a very small amount, as follows. Let $\TT$ be a tree on $n$ vertices, let $w \colon V(\TT) \rightarrow \Rnn$, and let $\varepsilon > 0$. We define the positive weight function $w_\varepsilon \colon V(\TT) \rightarrow \Rp$ on $G$ as follows:
\begin{align*}
	w_\varepsilon(v) = \begin{cases}
		w(v), & \text{ if } w(v) \neq 0\\
		\varepsilon, & \text{ otherwise.}
	\end{cases}
\end{align*}

We now prove that a centroid tree of $(\TT,w_\varepsilon)$ is also a centroid tree of $(\TT,w)$ if $\varepsilon$ is small enough. For this, it is enough to show that for every subgraph, no new centroids are introduced.

\begin{lemma}\label{centroids-zero-weights}
	Let $\TT$ be a tree on $n$ vertices, and let $w \colon V(\TT) \rightarrow \Rnn$.
	For each small enough $\varepsilon > 0$, each centroid of $(\TT,w_\varepsilon)$ is a centroid of $(\TT, w)$.
\end{lemma}
\begin{proof}
	Let $k < n$ be the number of zero-weight vertices w.r.t.\ $w$.
	Let $c$ be a centroid of $(\TT,g)$ and let $C \in \mathbb{C}(\TT-c)$. We have
	\begin{align*}
		w(C) \le w_\varepsilon(C) \le \frac{1}{2} w_\varepsilon(\TT) = \frac{1}{2} \left(w(\TT) + \varepsilon k\right) < \frac{1}{2} w(\TT) + \frac{1}{2} \varepsilon n.
	\end{align*}
	If $\varepsilon$ is small enough, this implies that $w(C) \le \frac{1}{2}w(\TT)$. Repeating the argument for each $C \in \mathbb{C}(\TT-c)$ shows that $c$ is a centroid of $(\TT,w)$.
\end{proof}

Note that if all weights are integers, we can simply set $\varepsilon = 1/n$ (or set $\varepsilon = 1$ and multiply each other weight by $n$). In general, we can treat $\varepsilon$ symbolically, without explicitly computing a value for it. 

We now show that the height of the centroid tree w.r.t.\ $w_\varepsilon$ is essentially bounded by the spread of $w$; in other words, replacing $w$ with $w_\varepsilon$ ensures that the computed centroid tree is ``reasonable''. Note that we can ignore the spread of $w_\varepsilon$ here (which may be very large, if $\varepsilon$ is very small).

\begin{lemma}\label{height-epsilon-spread}
	Let $\TT$ be a tree on $n$ vertices, let $w \colon V(\TT) \rightarrow \Rnn$, and let $\varepsilon > 0$ be defined as in \cref{centroids-zero-weights}. Each centroid tree of $(\TT,w_\varepsilon)$ has height $\fO( \log \sigma_w + \log n )$.
\end{lemma}
\begin{proof}
	Let $T$ be a centroid tree of $(\TT,w_\varepsilon)$. By \cref{centroids-zero-weights}, $T$ is also a centroid tree of $(\TT,w)$. Thus, the depth of each node $v$ with $w(v) > 0$ is at most $1 + \log \sigma_w$. Now consider a subtree $T_u$ with $w(T_u) = 0$. Then the weight w.r.t.\ $w_\varepsilon$ of each node in $T_u$ is the same, so the spread of $w_\varepsilon$
	restricted to $V(T_u)$ is $|V(T_u)| \le n$. This implies that the height of each zero-weight vertex in $T$ is at most $2 + \log \sigma_w + \log n$, concluding the proof.
\end{proof}

\Cref{centroids-zero-weights,height-epsilon-spread} imply that we can replace $w$ with $w_{\varepsilon}$ before running the algorithm. This yields

\begin{theorem} [Restatement of Thm.~\ref{p:construct-height}]
\label{thmnew}
	Let $\TT$ be a tree on $n$ vertices and $w$ be a weight function. We can compute a centroid tree of $(\TT,w)$ in time $\fO( n \log h ) \subseteq \fO( n \log \log {(\sigma+n)} )$, where $h$ is the height of the computed centroid tree and $\sigma$ is the spread of $w$.
\end{theorem}

\paragraph{Handling many zero-weight vertices.} 

If almost all vertices have weight zero, 
then \cref{thmnew} can be improved (in terms of $\sigma$) 
by a preprocessing stage. We briefly sketch the argument.

Let $\TT$ be the given tree, and $w$ be the given weight function. Let $m$ be the number of vertices with positive weight, and note that $\sigma_w \ge m$. We first transform $\TT$ by progressively removing zero-weight leaves, and replacing zero-weight degree-two vertices with an edge. This can be done with a single traversal in $\fO(n)$ time. In the resulting tree  $\TT'$, each zero-weight vertex has degree at least three (in particular, each leaf has positive weight). This means that at least half of the vertices in $\TT'$ have positive weight, so $|V(\TT')| \le 2m$. Computing a centroid tree $T$ on $\TT'$ thus requires $\fO(m \log \log (\sigma_w+m)) = \fO(m \log \log \sigma_w)$ time with \cref{p:construct-height} or $\fO(m \log m)$ with \cref{p:construct-general}. The latter bound is better when $\sigma_w \ge 2^m$.

It remains to show how to add the removed zero-weight vertices to the centroid tree $T$. Consider a maximal subgraph $\HH$ of removed vertices. First compute a rooting $T_\HH$ of $\HH$; since all vertices in $\HH$ have weight zero, $T_\HH$ is a centroid tree on $\HH$. Note that $\HH$ has up to two neighbors in $\TT$, all with positive weight. Let $v$ be the neighbor of $\HH$ that is farthest from the root in $T$, and attach $T_\HH$ to $T$ as a child of $v$. It is not hard to see that this requires $\fO(n)$ time in total, and that the resulting search tree is a centroid tree of $(\TT,w)$. Thus, we have
\begin{restatable}{theorem}{restateConstructLowSpread}\label{p:construct-low-spread}
	Let $\TT$ be a tree on $n$ vertices and $w$ be a weight function. We can compute a centroid tree of $(\TT, w)$ in time $\fO( n + m \log \log \sigma )$ or $\fO(n + m \log m)$, where $m$ is the number of vertices in $\TT$ with positive weight and $\sigma$ is the spread of $w$.
\end{restatable}

In particular, since $m \leq \sigma$, the running time becomes $\fO(n)$ when
$\sigma \in \fO(n/\log\log{n})$.

\subsection{Lower bounds on the running time}
We start with our generic construction for a lower bound.

\begin{figure}
\centering
\includegraphics[width=3.5in]{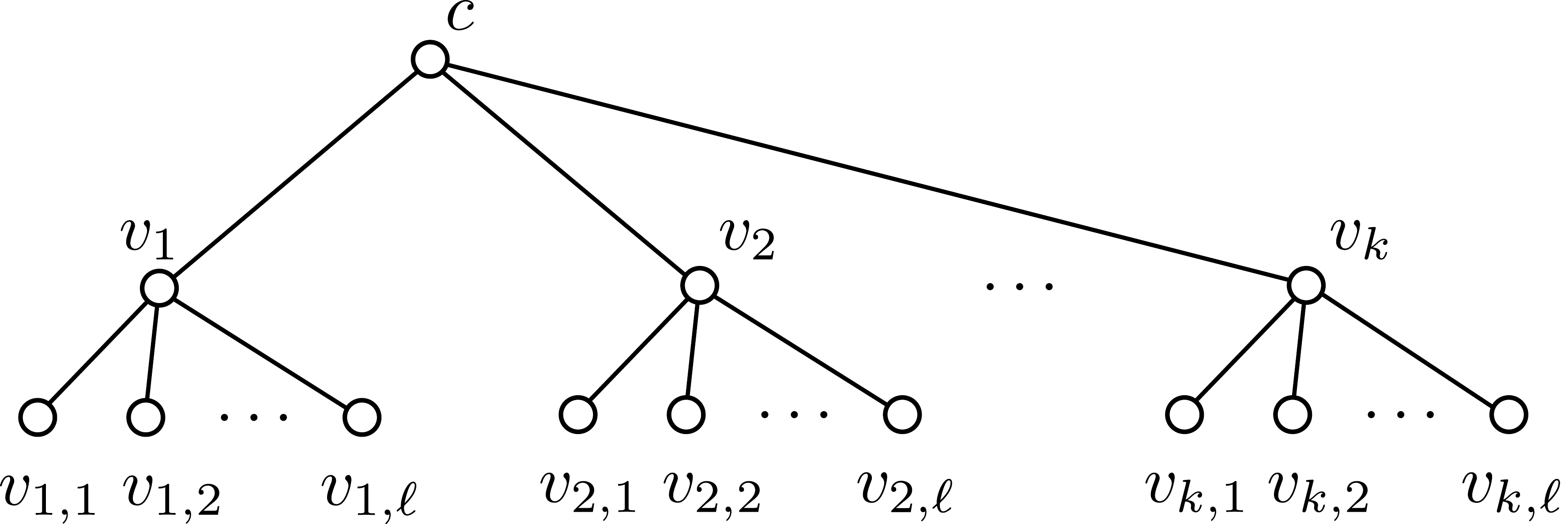}
\caption{Tree $\TT_{k,\ell}$ in the proof of Lemma~\ref{p:lb-construction}.}
\label{fig:lba}
\end{figure}

\begin{lemma}\label{p:lb-construction}
	Let $k \ge 1 $ and $\ell \ge 1$ be integers. There is a tree $\TT_{k,\ell}$ on $k \cdot (\ell+1) + 1$ vertices and a class $\fW_{k,\ell}$ of weight functions on $V(\TT_{k,\ell})$ such that
	\begin{enumerate}[(i)]
		\itemsep0pt
		\item $|\fW_{k,\ell}| = (\ell!)^k$;
		\item each $w \in \fW_{k,\ell}$ has spread $\sigma_w \le k \cdot 2^{\ell+1}$;
		\item for each $w \in \fW_{k,\ell}$, there is a unique centroid tree $T_w$ of $(\TT,w)$, and $T_w$ has height at least $\ell+1$;
		\item for each pair of distinct $w, w' \in \fW_{k,\ell}$, we have $T_w \neq T_{w'}$.
	\end{enumerate}
\end{lemma}
\begin{proof}
	Suppose first that $k \ge 2$. Let $\TT_{k,\ell}$ consist of a vertex $c$, and $k$ stars of size $\ell+1$, with centers $v_1, \dots, v_{k}$ adjacent to $c$. The remaining $\ell$ vertices in the star with center $v_i$ are denoted $v_{i,1}, \dots, v_{i,\ell}$, for all $1 \leq i \leq k$. See Figure~\ref{fig:lba}.
	
	Let $S_\ell$ denote the family of permutations of $\{1,\dots,\ell\}$. For permutations $\pi_1, \dots, \pi_{k} \in S_\ell$, let $w = w_{\pi_1, \dots, \pi_k}$ denote the weight function defined as $w(v_{i,j}) = 2^{\pi_i(j)}$,
	for all $1 \leq i \leq k$ and $1 \leq j \leq \ell$. In words, $w_{\pi_1, \dots, \pi_k}$ assigns the weights $2^1,\dots,2^{\ell}$ to the non-central vertices of the $i$-th star, permuted according to $\pi_i$, for all $i$. For the remaining vertices, $w(v_1) = \cdots = w(v_k) = 0$, and $w(c) = 1$. Let $\fW_{k,\ell}$ be the family of all such weight functions. Observe that $|\fW_{k,\ell}| = (\ell!)^k$ and the spread of each weight function is $k \cdot (2^{\ell+1}-1) + 1$, so (i) and (ii) hold.

	We claim that for any weight function $w \in \fW$, the centroid tree of $(\TT_{k,\ell},w)$ is unique, i.e., (iv) holds. 
	Indeed, let $w = w_{\pi_1, \dots, \pi_k}$, and observe that the unique centroid of $\TT_{k,\ell}$ is $c$ (since we assumed $k \ge 2$). The removal of $c$ splits $\TT_{k,\ell}$ into the $k$ stars with centers $v_1, \dots, v_k$.
	For all $1 \leq i \leq k$, the centroid tree of the star with center $v_i$ is uniquely determined by the weight assignment by the permutation $\pi_i$, and it is easily seen to be the path on vertices $v_{i,1}, \dots, v_{i,\ell}$ in decreasing order of weights, followed by the star center $v_i$. Finally, the entire search tree has height $\ell+2$, thereby (iii) holds.
	
	In the case $k = 1$, we omit the vertex $c$ and directly build a single star. Claims (i)--(iv) are easy to verify if we build $\fW_{k,\ell}$ as above.
\end{proof}

We are now ready to show our lower bounds. 

\restatethmq*

\begin{proof}
	We use \cref{p:lb-construction} with $k = \lfloor \frac{n-1}{h} \rfloor$ and $\ell = h-1$. The number of leaves in the decision tree is $|\fW_{k,\ell}|$, so its height is $\log |\fW_{k,\ell}| = \log ((\ell!)^k) \in \Omega( \frac{n}{h} \log (h!) ) = \Omega( n \log h)$.
\end{proof}

\Cref{thm:lb} implies that \cref{p:construct-height} is tight (up to a constant factor) for all $n$ and $h$.

A slight adaptation of the argument yields the following:

\begin{restatable}{theorem}{restatethmz}\label{p:lb-spread}
	Let $n \in \mathbb{N}$ and $\sigma \in \R$ with $4n \le \sigma \le 2^n$. Then there is a tree $\TT$ on $n$ vertices and a class $\fW$ of weight functions on $V(\TT)$ with spread at most $\sigma$, such that every binary decision tree that solves $\TT$ for $\fW$ has height $\Omega(n \log\log (\frac{\sigma}{n}) )$.
\end{restatable}

\begin{proof}
	We use \cref{p:lb-construction} with
	$\ell = \lfloor \log \sigma - \log n \rfloor \ge 2$ and
	$k = \lfloor \frac{n}{\ell} \rfloor \ge 1$. 
	The spread of each $w \in \fW_{k,\ell}$ is at most $\frac{n}{\ell} \cdot 2\frac{\sigma}{n} \le \sigma$, as desired. The height of each decision tree is $\log |\fW_{k,\ell}| \in \Omega( \ell k \log \ell ) = \Omega( n \log\log( \frac{\sigma}{n} ) )$.
\end{proof}

We now discuss for which range of the parameter $\sigma$ the bounds in Theorems~\ref{p:construct-low-spread} and \ref{p:lb-spread} are tight. First note that any reasonable model of computation will require $\Omega(n)$ time to read the input. Since we can pad the tree in the lower bound construction of Theorem~\ref{thm:lb} with zero-weight leaves, we can get a lower bound of $\Omega(n+ m \log\log (\frac{\sigma}{m}) )$ 
for all $m \leq n$, where $m$ is the number of positive-weight vertices.

Note that $\sigma \geq m$. If $\sigma > 2^m$, we can use \cref{p:lb-spread} with $\sigma = 2^m$, and obtain a tight $\Theta(n + m \log m)$ bound together with \cref{p:construct-low-spread}.

If $\sigma \le 2^m$ and $\sigma \ge m \cdot 2^{\log^{\varepsilon} m}$ for some constant $\varepsilon > 0$ (i.e., $\sigma$ is very slightly superlinear in $m$), then $\log\log (\frac{\sigma}{m}) \in \Theta( \log \log \sigma)$, so we obtain a tight $\Theta(n + m \log \log \sigma )$ bound from \cref{p:lb-spread} and \cref{p:construct-low-spread}.

If $\sigma$ is close to linear in $m$, then our upper bound $\fO(n + m \log\log \sigma)$ and our lower bound $\Omega(n + m \log\log(\frac{\sigma}{m}))$ differ. We leave the problem of finding the correct bounds in this case as an open question.

\section{Approximation guarantees of \texorpdfstring{$\al$}{alpha}-centroid trees}
\label{sec6}

\subsection{Upper bounds}

\restatethmg*

\paragraph{Part (i).}
The proof of Theorem~\ref{thm7}(i) is a generalization of the proof of \Cref{thm1}. We start with a lemma.

\begin{lemma}
\label{lemma99}
    Let $c$ be a $\al$-centroid of $(\TT,w)$ and $m=w(\TT)$. Then
    \begin{equation*}
        \OPT(\TT,w) \geq (1-\al)m + \al w(c) + \sum_{\mathclap{\HH\in\CC(\TT-c)}}~\OPT(\HH,w).
    \end{equation*}
\end{lemma}

\begin{proof}
    Let $T$ a search tree on $\TT$. Denote $r=\root(T)$. If $r=c$, the required follows using \Cref{observation99}. We assume therefore that $r\neq c$. Let $\HH^*\in\CC(\TT-c)$ such that $r\in V(\HH^*)$. Repeating the argument of \Cref{lemma1}, we get
    \begin{align*}
        \cost_w(T) &\geq m - w(\HH^*) + w(c) + \sum_\HH\OPT(\HH,w)\\
        &\geq (1-\al)m + w(c) + \sum_\HH\OPT(\HH,w),
    \end{align*}
    where the second inequality follows since $c$ is a $\alpha$-centroid.
\end{proof}

\begin{proof}[Proof of \Cref{thm7}(i)]
    By induction on the number of vertices. When $|V(\TT)|=1$ we have
    \begin{equation*}
        \frac{1}{1-\al}\OPT(\TT,w) - \frac{\al}{1-\al}m =
        \frac{1}{1-\al}m - \frac{\al}{1-\al}m =
        m =
        \CENT^\al(\TT,w),
    \end{equation*}
    as required.
    
    Assume $|V(\TT)|>1$. Repeating the argument of \Cref{thm1}, using the induction hypothesis we get
    \begin{align*}
        \CENT^\al(\TT,w) &\leq m - \frac{\al}{1-\al}(m-w(c)) + \frac{1}{1-\al}\sum_{\mathclap{\HH\in\CC(\TT-c)}}~\OPT(\HH,w)\\
        &\leq \frac{1}{1-\al}\OPT(\TT,w) - \frac{\al}{1-\al}m,
    \end{align*}
    where the last inequality is exactly \Cref{lemma99}.
\end{proof}

\paragraph{Part (ii).} We prove Theorem~\ref{thm7}(ii) through the following lemma. 

\begin{lemma}
\label{lemma999}
    Let $\TT$ be a tree, $w:V(\TT) \rightarrow \mathbb{R}_{\geq 0}$, $m=w(\TT)$ and $\al\in[\frac{1}{3},\frac{1}{2}]$. Assume that $c$ is an $\al$-centroid. Then
    \begin{equation}
    \label{eq8}
        \OPT(\TT,w) \geq (2-3\al)m + (3\al-1)w(c) + \sum_{\mathclap{\HH\in\CC(\TT-c)}}~\OPT(\HH,w).
    \end{equation}
\end{lemma}

\begin{proof}
    Let $T$ be an arbitrary search tree on $\TT$. We will show that $\cost_w(T)$ is at least the right hand side of \Cref{eq8}.
    
    Denote $r=\root(T)$. If $r=c$, the required follows using \Cref{observation99}, observing that $m\geq (2-3\al)m + (3\al-1)w(c)$. Assume that $r\neq c$. Let $\HH_0\in\CC(\TT-c)$ such that $r\in V(\HH_0)$ and denote $s=\LCA_T(\TT\setminus \HH_0)\in V(\TT-\HH_0)$. Consider two cases:
    
    \paragraph{Case $s=c$.} The contribution of vertices of $\HH_0$ to $\cost_w(T)$ is at least $\cost_w(T|_{\HH_0})$. For $\HH\in\CC(\TT-c)$, $\HH\neq \HH_0$, the contribution of vertices of $\HH$ is at least $2w(\HH) + \cost_w(T|_\HH)$, since these vertices have both $c$ and $r$ as predecessors. The contribution of $c$ is at least $2w(c)$. Summing all the above, we get
    \begin{align*}
        \cost_w(T) &\geq 2w(c) + \cost_w(T|_{\HH_0}) + \sum_{\HH\neq \HH_0}(2w(\HH) + \cost_w(T|_\HH))\\
        &\geq 2(m-w(\HH_0)) + \sum_\HH\OPT(\HH,w)\\
        &\geq m + \sum_\HH\OPT(\HH,w),
    \end{align*}
    where the last inequality follows from $c$ being an $\al$-centroid and $\al\leq \frac{1}{2}$.
    
    \paragraph{Case $s\neq c$.} Let $\HH_1\in\CC(\TT-c)$ such that $s\in V(\HH_1)$. Note that $\HH_1\neq \HH_0$. The contribution of vertices of $\HH_1$ to $\cost_w(T)$ is at least $w(\HH_1) + \cost_w(T|_{\HH_1})$, since these have $r$ as predecessor. The contribution of $c$ is at least $3w(c)$, since it has both $r$ and $s$ as predecessors. The contribution of all other vertices is bounded from below as in the previous case. Summing up, we get
    \begin{align*}
        \cost_w(T) &\geq 3w(c) + w(\HH_1) + 2\sum_{\mathclap{\HH\neq \HH_0,\HH_1}}~w(\HH) + \sum_\HH\cost_w(T|_\HH)\\
        &\geq 2m - 2w(\HH_0) - w(\HH_1) + w(c) + \sum_\HH\OPT(\HH,w)\\
        &\geq (2-3\al)m + w(c) + \sum_\HH\OPT(\HH,w),
    \end{align*}
    where the last inequality follows from $c$ being an $\al$-centroid. (See \Cref{alpha_lemma}.)
\end{proof}

\begin{figure}
    \centering
    \includegraphics[width=2in]{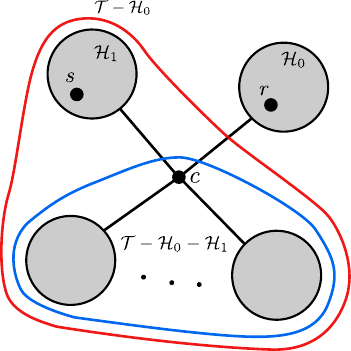}
    \caption{Illustration of case $s\neq c$ in the proof of \Cref{lemma999}. Vertices in $\TT-\HH_0$ have both $r$ and $s$ as ancestors.}
    \label{alpha_lemma}
\end{figure}

\begin{proof}[Proof of \Cref{thm7}(ii)]
    By induction on the number of vertices. When $|V(\TT)|=1$ we have
    \begin{equation*}
        \frac{1}{2-2\al}\OPT(\TT,w) - \frac{3\al-1}{2-3\al}m =
        \frac{1}{2-3\al}m - \frac{3\al-1}{2-3\al}m =
        m \geq
        \CENT^\al(\TT,w),
    \end{equation*}
    as required.
    
    Assume $|V(\TT)|>1$. Repeating the argument of \Cref{thm1}, using the induction hypothesis we get
    \begin{align*}
        \CENT^\al(\TT,w) &\leq m - \frac{3\al-1}{2-3\al}(m-w(c)) + \frac{1}{2-3\al}~~~\sum_{\mathclap{\HH\in\CC(\TT-c)}}~\OPT(\HH,w)\\
        &\leq \frac{1}{2-3\al}\OPT(\TT,w) - \frac{3\al-1}{2-3\al}m,
    \end{align*}
    where the last inequality is exactly \Cref{lemma999}.
\end{proof}

\subsection{Lower bound}

We show that \Cref{thm7}(i) is tight when $\al\geq\frac{1}{2}$.

\restatethmh*

\begin{proof}
    Let $\TT_n$ be defined recursively as in \Cref{thm2} and $w_n$ is defined recursively as follows. Let $\A$, $\B$ and $c$ be as in the definition of $\TT_n$. Let $w_\A$ and $w_\B$ be weight functions on $\A$ and $\B$ respectively, each a copy of $w_{n-1}$. Then $w_n$ is given by
    \begin{equation*}
    w_n(v) = \begin{cases}
            0, & v=c\\
			\al w_\A(v), & v\in V(\A)\\
			(1-\al)w_\B(v), & v\in V(\B).
            \end{cases}
    \end{equation*}
    
    Let $C_n$ and $T_n$ be the same search trees on $\TT_n$ as in the proof of \Cref{thm2}. Observe that $C_n$ is an $\al$-centroid tree on $\TT_n$. With a similar analysis as  in \Cref{thm2}, we get $\cost_{w_n}(C_n)=n+1$ and $\cost_{w_n}(T_n)=(1-\al)n+1$. From these the required bound follows. Using \Cref{thm7}(i), we have $\displaystyle\lim_{n\rightarrow\infty}\OPT(\TT_n,w_n)=\infty$.
\end{proof}

\subsection{Optimal STTs}

In this section we prove the characterization of optimal STTs as $\alpha$-centroid trees.

\restatethmoptc*

For the purpose of proving \Cref{thm10}, we define the following transformation on search trees. Let $T$ be a search tree on $\TT$ and $v\in V(\TT)$. We denote by $T^v$ the search tree on $\TT$ that is obtained as follows. Set $\root(T^v)=v$. The subtrees of $T^v$ rooted at the children of $v$ are all the trees of the form $T|_\HH$, where $\HH\in\CC(\TT-v)$. We say that $T^v$ is obtained from $T$ by \emph{lifting} the vertex $v$. (This transformation can also be defined via \emph{rotations}, e.g., see~\cite{Bose20}. $T^v$ is obtained by repeatedly rotating the edge between $v$ and its parent until $v$ becomes the root.) For every $u,v\in V(\TT)$ we have
\begin{equation*}
   \Path_{T^v}(u) =
        \begin{cases}
            \{v\}, & u=v\\
    		\{v\}\cup\left(\Path_T(u)\cap V(\HH)\right), &u\in V(\HH),~\HH\in\CC(\TT-v).
        \end{cases} 
\end{equation*}

\begin{proof}[Proof of \Cref{thm10}]
    Let $T$ be an optimal STT on $\TT$, and suppose towards contradiction, that $T$ is not a $\frac{2}{3}$-centroid. Assume w.l.o.g., that $w(\TT)=1$. By taking $T$ to be a minimum height counterexample, we can assume that the root $x$ of $T$ has a child $y$, so that $w(T_y) > \frac{2}{3}$.  Let $B_x$ denote the set of vertices \emph{not} in the same component of $\TT-x$ as $y$, and let $B_y$ denote the set of vertices not in the same component of $\TT-y$ as $x$. Finally, let $B_{x,y}=V(T)-B_x-B_y$. By our assumption, $w(B_y\cup B_{x,y})>\frac{2}{3}$, and thus, $w(B_x) < \frac{1}{3}$. We distinguish three cases.
 
 \begin{figure}[h]
\centering
\includegraphics[width=5in]{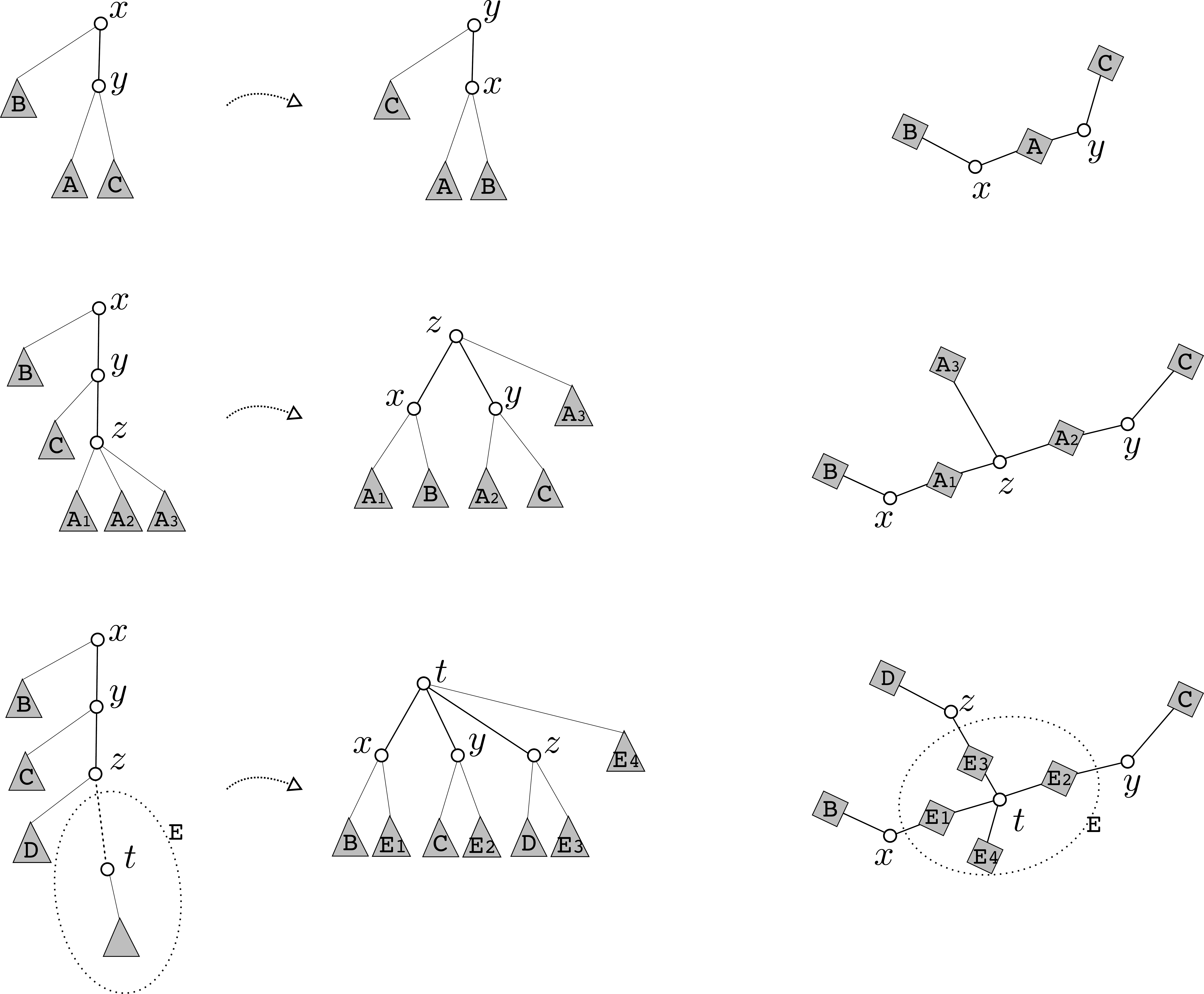}
\caption{Illustration of the proof of Theorem~\ref{thm10}, Case 1--3, top to bottom. (\emph{Left.}) Transformation of search tree $T$. (\emph{Right.}) Underlying tree $\TT$. Note that some blocks may indicate multiple disjoint subtrees.}
\label{opt23}
\end{figure}

    \paragraph{Case 1.} If $w(B_y)>\frac{1}{3}$, then transform $T$ by lifting $y$. The depths of vertices in $B_y$ decrease by one and the depths of vertices in $B_x$ increase by one. The depths of vertices in $B_{x,y}$ are unchanged. (See Figure~\ref{opt23}.1.) We have
    \begin{equation*}
        \cost_w(T^y) - \cost_w(T) = w(B_x) - w(B_y) < \frac{1}{3} - \frac{1}{3} = 0, 
    \end{equation*}
    contradicting the optimality of $T$.
    
    If Case 1 did not occur, we have $w(B_{x,y})>\frac{1}{3}$, and in particular, $B_{x,y}\neq\emptyset$. Denote by $z$ the unique child of $y$ in $B_{x,y}$.

    \paragraph{Case 2.} Assume $z$ is in the path in $\TT$ between $x$ and $y$. Then, transform $T$ by lifting $z$. The depths of vertices in $B_x$ increase by one, since each of these vertices gains $z$ as ancestor. The depths of vertices in $B_y$ are unchanged, since these vertices gain $z$ and lose $x$ as ancestors. The depths of vertices in $B_{x,y}$ are decreased by \emph{at least} one, since each loses at least one ancestor from $\{x,y\}$. (See Figure~\ref{opt23}.2.) We then have
    \begin{equation*}
        \cost_w(T^z) - \cost_w(T) \leq w(B_x) - w(B_{x,y}) < \frac{1}{3} - \frac{1}{3} = 0,
    \end{equation*}
    again, a contradiction.
    
    \paragraph{Case 3.} Finally, assume that $z$ is not in the path between $x$ and $y$. Note that this cannot happen in the special case of BSTs, i.e., when $\TT$ is a path.
    
    Let $t\in B_{x,y}$ be the unique vertex in $\TT$ whose removal separates $x$, $y$ and $z$ to different components. Now, transform $T$ by lifting $t$. As in the previous case, depths of vertices in $B_x$ increase by one, depths of vertices in $B_y$ stay the same and depths of vertices $B_{x,y}$ decrease by \emph{at least} one. To see this, observe that each vertex in $B_{x,y}$ gains $t$ as ancestor and loses two ancestors from $\{x,y,z\}$, and possibly more. (See Figure~\ref{opt23}.3.) As in the previous case, we have $\cost_w(T^t)<\cost_w(T)$, a contradiction. 
\end{proof}

\section{Conclusions}\label{sec7}
We showed that the average search time in a centroid tree is larger by at most a factor of $2$ than the smallest possible average search time in an STT and that this bound is tight.
We also showed that centroid trees can be computed in  $\fO(n\log h)$ time 
where $h$ is the height of the centroid tree. 

Perhaps the most intriguing question is to determine whether the problem of computing an optimal STT is in {P}. A
secondary goal would be to achieve an approximation ratio better than $2$ in near linear time. (The running time of the STT's of Berendsohn and Kozma~\cite{BK22} degrade as $\fO(n^{2k+1})$ for a $\left( 1+\frac{1}{k} \right)$-approximation.) As for centroid trees, a remaining question is whether they can be computed in $\fO(n)$ time whenever the spread of the weight function is $\sigma \in \fO(n)$.

A special case in which high quality approximation can be efficiently found is when an $\al$-centroid tree exists for $\al<\frac{1}{2}$. This case can be recognized and handled in near linear time using our algorithm. (Observe that an $\al$-centroid tree for $\al<\frac{1}{2}$ is also the unique $\frac{1}{2}$-centroid tree.) \Cref{thm7}(ii) gives strong approximation guarantees for this case, yielding the \emph{optimum} when $\al\leq\frac{1}{3}$. It is an interesting question whether the bounds can be improved for $\al$ in the range $\left(\frac{1}{3},\frac{1}{2}\right)$, i.e., whether Theorem~\ref{thm7}(ii) is tight.

A small gap remains in the exact approximation ratio of centroid trees when $\TT$ has maximum degree $\Delta$ and $\OPT$ is unbounded, i.e., between the upper bound $(2-\frac{1}{2^{\Delta}})$ of Theorem~\ref{thm4} and the lower bound $(2-\frac{4}{2^\Delta})$ of Theorem~\ref{thm5}(ii).

\newpage
\noindent
{\bf \Large Appendix}

\appendix

\section{Projection of a search tree}\label{appa}

\restateDefProjection*

\begin{proof}
    We denote by $u<_Tv$ that $v$ is an ancestor of $u$ in $T$. We need to show that there is a unique search tree $T|_\HH$ on $\HH$ such that for all $u,v\in V(\HH)$,
    \begin{equation}
    \label{eq999}
        u<_{T|_\HH}v \iff u<_Tv.
    \end{equation}
    
    $T|_\HH$ is constructed recursively as follows. Denote $s=\LCA_T(\HH)\in V(\HH)$. Set $s$ as the root of $T|_\HH$. The subtrees of $T|_\HH$ rooted at the children of $s$ will be all the trees $T|_C$, where $C\in\CC(\HH-s)$.
    
    That $T|_\HH$ is a search tree on $\HH$ is easily shown by induction on $|V(\HH)|$. Observing that (\ref{eq999}) implies $\root(T|_\HH)=\LCA_T(\HH)$, uniqueness also follows easily by induction. We show here by induction that (\ref{eq999}) holds. Let $u,v\in V(\HH)$ and let $r=\mathrm{root}(T|_\HH)$. If $u$ and $v$ are in the same connected component $C\in\CC(\HH-r)$, then by induction
    \begin{equation*}
        u<_{T|_\HH}v \iff u<_{T|_C}v \iff u<_T v,
    \end{equation*}
    as required. If $u$ and $v$ are in two different connected components of $\HH-r$, then by construction, $u\not<_{T|_\HH}v$. In that case since $r=\LCA_T(\HH)$, we have $u<_Tr$ and $v<_Tr$, therefore also $u\not<_Tv$, as required. If $v=r$ then both $u<_{T|_\HH}v$ and $u<_Tv$ hold. Finally, if $u=r$ then both $u\not<_{T|_\HH}v$ and $u\not<_Tv$ hold.
\end{proof}

We repeat here the definition of projection of a search tree given by Cardinal et al.~\cite{Cardinal18}. To avoid confusion, we denote by $\tilde{T}|_\HH$ the projection of $T$ to $\HH$ as defined in~\cite{Cardinal18}, and by $T|_\HH$ the projection as defined in \Cref{projection_def}. We than show that in fact, $\tilde{T}|_\HH=T|_\HH$.

As a preliminary step, Cardinal et al.\ show how to construct the projection of $T$ to $\TT-x$, where $x$ is a leaf of $\TT$. They distinguish between three cases:
\begin{enumerate}
    \item $x$ has a parent and no child in $T$,
    \item $x$ has a parent and a single child in $T$,
    \item $x=\root(T)$.
\end{enumerate}
In the first case, $x$ is simply removed from $T$. In the second case, $x$ is removed from $T$ and an edge is added between its parent and its child. In the third case, $x$ is removed and its only child is selected as the new root. This operation is referred to as \emph{pruning} $x$ from $T$. The projection $\tilde{T}|_\HH$ is then defined as the tree obtained by iteratively pruning leaves that are not in $\HH$.

It is straightforward to verify that the pruning operation respects the equivalence (\ref{eq999}). More precisely, if $T'$ is obtained from $T$ by pruning $x$, then $u \leq_{T'} v \iff u \leq_T v$ whenever $u,v\neq x$. By induction on the number of pruned vertices, it follows that
\begin{equation*}
    u \leq_{\tilde{T}|_\HH} v \iff u \leq_T v
\end{equation*}
for all $u,v\in V(\HH)$. That $\tilde{T}|_\HH=T|_\HH$ follows from the uniqueness in \Cref{projection_def}.

\section{Centroid and median}\label{appb2}

\begin{figure}
	\centering
	\includegraphics[width=2in]{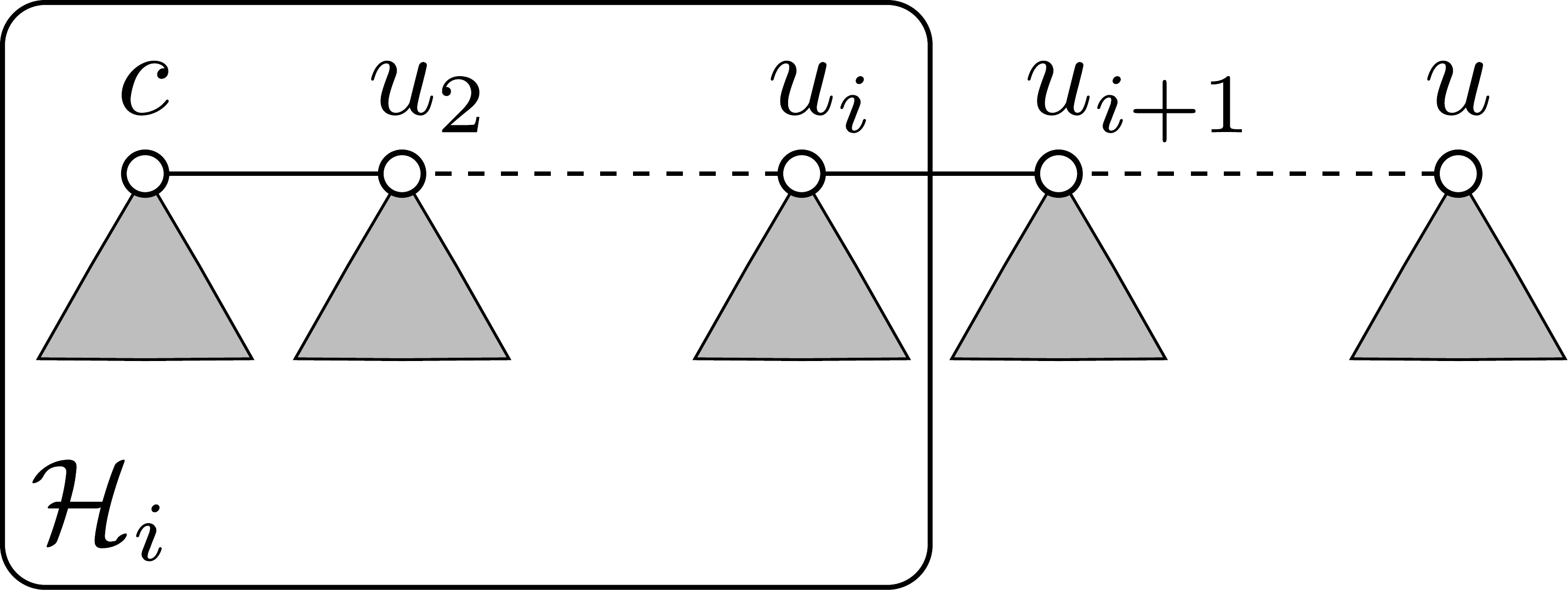}
	\caption{An illustration of the second part of \cref{p:centroid-is-median}.} \label{fig:centroid-median}
\end{figure}

\lemrestatecm*

\begin{proof}
	Let $u \in V(\TT)$ minimize $W(u)$, and suppose there is a component $\HH \in \CC(\TT-u)$ with $w(\HH) > \frac{1}{2}w(\TT)$. Let $v$ be the neighbor of $u$ in $\HH$.
	For each $x \in V(\TT)$, we have $\dist_\TT(v,x) = \dist_\TT(u,x) - 1$ if $x \in V(\HH)$, and $\dist_\TT(v,x) = \dist_\TT(u,x) - 1$ otherwise. Thus, $W(v) = W(u) + w(\TT-\HH) - w(\HH) < W(u)$, a contradiction.

	On the other hand, let $c$ be a centroid of $(\TT,w)$, and let $u \in V(\TT)-\{c\}$. We show that $W(u) \ge W(c)$.
	Let $c = u_1, u_2, \dots, u_k = u$ be the path from $c$ to $u$ in $\TT$. We claim that $W(u_i) \le W(u_{i+1})$ for all $i \in [k-1]$. This in particular implies $W(c) \le W(u)$, as desired.
	
	Towards our claim, let $i \in [k-1]$, and let $\HH_i$ be the component of $\TT - u_{i+1}$ that contains $u_i$ (see \cref{fig:centroid-median}). Note that $\HH_i$ encompasses $c$ and all but one component of $\CC(\TT-c)$, so $w(\HH_i) \ge \frac{1}{2} w(\TT)$. For each $x \in V(\TT)$, we have $\dist_\TT(u_i,x) = \dist_\TT(u_{i+1},x) - 1$ if $x \in V(\HH_i)$, and $\dist_\TT(v,u_i) = \dist_\TT(u,u_{i+1}) + 1$ otherwise. Thus, $W(u_i) - W(u_{i+1}) + W(\HH_i) + W(\TT-\HH_i) \le W(u_{i+1})$. This concludes the proof
	\end{proof}
	
\section{Tie-breaking}\label{appc}
	
	\lemrestateq*
	
	We prove \Cref{tie_breaking_lemma} by the following series of simple observations. 
\begin{observation}
\label{observation1}
    Let $w_1,w_2:V(\TT)\rightarrow \mathbb{R}_{\geq0}$ be two weight functions. Assume that a vertex $c\in V(\TT)$ is a centroid of both $(\TT,w_1)$ and $(\TT,w_2)$. Assume further that $c$ is the unique centroid of $(\TT,w_2)$. Then $c$ is the unique centroid of $(\TT,w_1+w_2)$.
\end{observation}

\begin{proof}
    Denote $W_i(u)=\sum_{v\in V(\TT)}\dist_\TT(u,v)\cdot w_i(v)$. If $c$ is a minimum of $W_1$ and is the unique minimum of $W_2$, then it is the unique minimum of $W_1+W_2$, and by Lemma~\ref{p:centroid-is-median} it is a unique centroid of $(\TT,w_1+w_2)$.
\end{proof}

\begin{observation}
\label{observation2}
    Let $w:V(\TT)\rightarrow \mathbb{R}_{\geq0}$ and $m=w(\TT)$. If $w(u)\geq m/2$ then $u$ is a centroid of $(\TT,w)$. Moreover, if $w(u)>m/2$, then $u$ is the unique centroid of $(\TT,w)$.
\end{observation}

\begin{proof}
    $u$ is a centroid since the total weight of $\TT-u$ is $\leq m/2$. Assume $w(u)>m/2$ and let $v\neq u$. The connected component of $\TT-v$ in which $u$ is has weight $>m/2$, therefore $v$ is not a centroid.
\end{proof}

\begin{observation}
\label{observation3}
    Let $T$ be a search tree on $\TT$. There exists a $w:V(\TT)\rightarrow \mathbb{R}_{\geq0}$ such that $T$ is the unique centroid tree of $(\TT,w)$.
\end{observation}

\begin{proof}
    By induction on $n$. When $n=1$ every $w$ has the desired property. Assume $n>1$. Let $r=\root(T)$ and let $T_1,\dots,T_d$ be the subtrees of $T$ rooted at the children of $r$. Let $\TT_i=\TT[V(T_i)]$. By the induction hypothesis there are $w_i:V(\TT_i)\rightarrow \mathbb{R}_{\geq0}$ such that $T_i$ is the unique centroid tree of $(\TT_i,w_i)$. Let $w:V(\TT)\rightarrow \mathbb{R}_{\geq0}$ be defined by
    \begin{equation*}
        w(u) =
        \begin{cases}
            w_i(u), & u\in V(\TT_i)\\
    		1 + \sum_{i=1}^dw_i(\TT_i), &u=r.
        \end{cases}
    \end{equation*}
    By \Cref{observation2} $r$ is the unique centroid of $(\TT,w)$.
\end{proof}

\begin{proof}[Proof of \Cref{tie_breaking_lemma}]
    Using \Cref{observation3}, let $\tilde{w}:V(\TT)\rightarrow \mathbb{R}_{\geq0}$ be such that $C$ is the unique centroid tree of $(\TT,\tilde{w})$. We can scale $\tilde{w}$ so that $\lVert \tilde{w}\rVert_\infty<\epsilon$. Denote $w'=w+\tilde{w}$. Using \Cref{observation1}, by induction on the height of $T$, it follows that $T$ is the unique centroid tree of $(\TT,w')$.
\end{proof}

	\newpage 
	
	\bibliography{main}{}
	\bibliographystyle{alpha}

\end{document}